\documentclass[conference]{IEEEtran}
\IEEEoverridecommandlockouts
\usepackage[dvipsnames]{xcolor}

\usepackage{cite}
\usepackage{amsmath,graphicx}
\usepackage{amssymb,amsfonts}
\usepackage{graphicx}
\usepackage{textcomp}
\usepackage{subfigure}
\usepackage{textcomp} 
\usepackage{siunitx}  
\usepackage{algorithm}
\usepackage{algpseudocode}
\usepackage{xpatch}
\usepackage{pgfplots}
\usepackage{pgfplotstable}
\usepackage{readarray}
\usepackage{siunitx}
\usepackage{bm}
\usepackage[nolist]{acronym}
\usepackage{bbm}
\usepackage{url}
\usepackage{amsthm}

\usetikzlibrary{spy}
\newcommand{\Imag}[1]{\mathrm{Im} #1}

\newcommand{\Real}[1]{\mathrm{Re} #1}
\newcommand{\abs}[1]{\left|#1\right|}
\newcommand{\eyeM}{\bm{\mathrm{I}}}

\newcommand{\bD}{\bm{D}}

\newcommand{\Her}{{\mathrm{H}}}

\newcommand{\bC}{\bm{C}}

\newcommand{\NB}{M}
\newcommand{\NM}{K}
\newcommand{\NRe}{N}

\newcommand{\bv}{\bm{v}}

\newcommand{\norm}[1]{\left\|#1\right\|}

\newcommand{\summe}[2]{\sum_{#1}^{#2}}

\newcommand{\teR}{{\text{R}}}

\newcommand{\teRS}{{\text{RS}}}
\newcommand{\teDS}{{\text{DS}}}
\newcommand{\teDR}{{\text{DR}}}

\newcommand{\teDN}{{\text{DeN}}}
\newcommand{\Dec}{{\text{DN}}}
\newcommand{\tL}{{\text{L}}}

\newcommand{\bi}{\bm{i}}
\newcommand{\bb}{\bm{b}}
\newcommand{\be}{\bm{e}}

\newcommand{\ba}{\bm{a}}

\newcommand{\bx}{\bm{x}}

\newcommand{\bTheta}{\bm{\Theta}}

\newcommand{\bZ}{\bm{Z}}

\newcommand{\bz}{\bm{z}}

\newcommand{\bzero}{\bm{0}}
\newcommand{\tG}{{\text{G}}}
\newcommand{\teL}{{\text{L}}}
\newcommand{\tN}{{\text{N}}}
\newcommand{\inv}{{-1}}
\newcommand{\transpo}{{\mathrm{T}}}
\newcommand{\im}{\mathrm{j}}

\newcommand{\diag}{\mathrm{diag}}
\newcommand{\He}{\mathrm{H}}
\newcommand{\ones}{\bm{1}}

\definecolor{TUMBeamerYellow}    {rgb} {1.000,0.706,0.000}    
\definecolor{TUMBeamerOrange}    {rgb} {1.000,0.502,0.000}    
\definecolor{TUMBeamerRed}       {rgb} {0.898,0.204,0.094}    
\definecolor{TUMBeamerDarkRed}   {rgb} {0.792,0.129,0.247}    
\definecolor{TUMBeamerBlue}      {rgb} {0.000,0.600,1.000}    
\definecolor{TUMBeamerLightBlue} {rgb} {0.255,0.745,1.000}    
\definecolor{TUMBeamerGreen}     {rgb} {0.569,0.675,0.420}    
\definecolor{TUMBeamerLightGreen}{rgb} {0.710,0.792,0.510}    

\newtheorem{theorem}{Theorem}

\newtheorem{proposition}{Proposition}
\newtheorem{corollary}{Corollary}

\begin{acronym}
    \acro{AoD}{angle of departure}
    \acro{AoA}{angle of arrival}
    \acro{ULA}{uniform linear array}
    \acro{CSI}{channel state information}
    \acro{LOS}{line of sight}
    \acro{EVD}{eigenvalue decomposition}
    \acro{BS}{base station}
    \acro{MS}{mobile station}
    \acro{mmWave}{millimeter wave}
    \acro{DPC}{dirty paper coding}
    \acro{RIS}{reconfigurable intelligent surface}
    \acro{AWGN}{additive white gaussian noise}
    \acro{MIMO}{multiple-input multiple-output}
    \acro{SISO}{single-input single-output}
    \acro{UL}{uplink}
    \acro{DL}{downlink}
    \acro{OFDM}{orthogonal frequency-division multiplexing}
    \acro{TDD}{time-division duplex}
    \acro{LS}{least squares}
    \acro{MMSE}{minimum mean square error}
    \acro{SINR}{signal to interference plus noise ratio}
    \acro{OBP}{optimal bilinear precoder}
    \acro{LMMSE}{linear minimum mean square error}
    \acro{MRT}{maximum ratio transmitting}
    \acro{M-OBP}{multi-cell optimal bilinear precoder}
    \acro{S-OBP}{single-cell optimal bilinear precoder}
    \acro{SNR}{signal-to-noise ratio}
    \acro{SDR}{Semidefinite Relaxation}
    \acro{SE}{spectral efficiency}
    \acro{GCEs}{Gram channel eigenvalues}
    \acro{LNA}{low-noise amplifier}
    \acro{Tx}{transmit side}
    \acro{Rx}{receive side}
    \acro{HPA}{high power amplifier}
    \acro{LNA}{low noise amplifier}
    \acro{STAR}{simultaneously transmitting and reflecting}
    \acro{BD}{beyond diagonal}
\end{acronym}

\def\BibTeX{{\rm B\kern-.05em{\sc i\kern-.025em b}\kern-.08em
    T\kern-.1667em\lower.7ex\hbox{E}\kern-.125emX}}
\begin{document}

\title{Decoupling Networks and Super-Quadratic Gains\\ for RIS Systems with Mutual Coupling
}

\author{\IEEEauthorblockN{Dominik Semmler, Josef A. Nossek, \textit{Life Fellow, IEEE}, Michael Joham,\\
Benedikt Böck, and Wolfgang Utschick, \textit{Fellow, IEEE}}
\thanks{Preliminary results have been presented at the ISWCS 2024~\cite{ConferenceDecoupling}.}
\IEEEauthorblockA{\textit{School of Computation, Information and Technology, Technical University of Munich, 80333 Munich, Germany} \\
email: \{dominik.semmler,josef.a.nossek,joham,benedikt.boeck,utschick\}@tum.de}}

\maketitle

\begin{figure}[b]
    \onecolumn
    \centering
    \scriptsize{This work has been submitted to the IEEE for possible publication. Copyright may be transferred without notice, after which this version may no longer be accessible.}
    \vspace{-1.3cm}
    \twocolumn
\end{figure}
\begin{abstract}
We propose decoupling networks for the \ac{RIS} array as a solution to benefit from the mutual coupling between the reflecting elements.
In particular, we show that when incorporating these networks, the system model reduces to the same structure as if no mutual coupling is present.
Hence, all algorithms and theoretical discussions neglecting mutual coupling can be directly applied when mutual coupling is present by utilizing our proposed decoupling networks.
For example, by including decoupling networks, the channel gain maximization in \ac{RIS}-aided \ac{SISO} systems does not require an iterative algorithm but is given in closed form as opposed to using no decoupling network.
In addition, this closed-form solution allows to analytically analyze scenarios under mutual coupling resulting in novel connections to the conventional transmit array gain.
In particular, we show that super-quadratic (up to quartic) channel gains w.r.t. the number of \ac{RIS} elements are possible and, therefore, the system with mutual coupling performs significantly better than the conventional uncoupled system in which only squared gains are possible.
We consider diagonal as well as \ac{BD}-\acp{RIS} and give various analytical and numerical results, including the inevitable losses at the \ac{RIS} array.
In addition, simulation results validate the superior performance of decoupling networks w.r.t. the channel gain compared to other state-of-the-art methods.
\end{abstract}


\section{Introduction}
\Acp{RIS} are currently highly discussed as they are considered an important technology for future wireless communications systems (see \cite{Power_Min_IRS}, \cite{SmartRadioEnvironment}).
\Acp{RIS} are arrays consisting of many passive reflecting elements which can manipulate the incoming wavefronts and, hence, shape the propagation environment.
The potential of \acp{RIS} has already been demonstrated and utilizing a \ac{RIS} has shown to significantly improve, e.g., power consumption \cite{Power_Min_IRS} and energy efficiency \cite{EnergyEff} in various scenarios.
In these publications each of the reflecting elements is modeled by a phase manipulation.
This model is the most popular and is used in the majority of the \ac{RIS} literature (see  \cite{WSR,WMMSEMIMO,Eigenvalues,MIMOP2PCap,ZeroForcLOS,StatisticalSadaf,StatisticalCSITwo,StatisticalCSITwoTimeScale}).

The architecture of \ac{RIS} arrays is still under research and various approaches are discussed,
including \ac{STAR} \acp{RIS} \cite{STARRIS,STARRISTwo}, active \acp{RIS} \cite{ActiveVSPassive}, \cite{ActiveVSPassiveTwo} as well as \ac{BD}-\acp{RIS} \cite{ScatteringRIS,BDRISOptimalSolution,BeyondDiagonal}.
For \ac{BD}-\acp{RIS}, the manipulation matrix at the \ac{RIS} is not restricted to be diagonal but is generalized to a non-diagonal matrix allowing the different elements to be connected with each other (see \cite{ScatteringRIS}).
This leads to an improved performance at the cost of having an increased complexity.
Regardless of the chosen architecture, the elements at the \ac{RIS} are typically mounted on a surface close to each other.
Hence, mutual coupling between the \ac{RIS} elements is inevitable in such a system.
The effect of coupling between the elements is usually neglected, however, the system performs very differently when mutual coupling is included.
When including mutual coupling, a more advanced \ac{RIS} model is needed in comparison to the phase-shift model.

\begin{figure}
    \centering
\includegraphics*{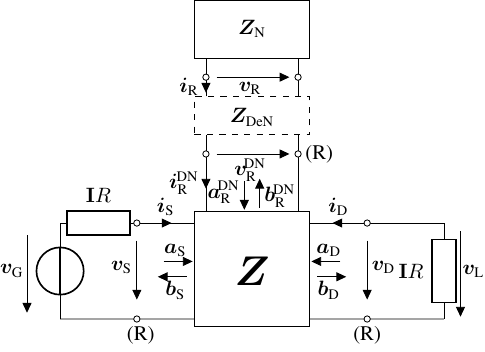}
\caption{Multiport Model of \cite{NewChannelModel} with an additional decoupling network $\bZ_\teDN$.}
\label{fig:Three_Port}
\end{figure}
Modeling the \ac{RIS} is under ongoing research and different models for \acp{RIS} have already been proposed in the literature.
In, e.g., \cite{MutualCouplingAware}, a \ac{RIS} model has been introduced that is based on impedance matrices.
Impedance-based descriptions have already been shown to provide a powerful description for conventional communication systems without the \ac{RIS} (see \cite{TowardCircuitTheory}).
The model in \cite{MutualCouplingAware} allows to include the effect of mutual coupling as well as different \ac{RIS} architectures.
Furthermore, in \cite{ScatteringRIS}, the \ac{RIS} has been modeled based on scattering parameters which can also include the effects of mutual coupling (see \cite{FollowUpBDRIS}).
Both models describe the same electromagnetic system and, hence, lead to the same conclusions.
A comparsion of these two models is given in \cite{NewChannelModel,FollowUPSZJournal,FollowUpUniversal,WSAComparison} and with the correct interpretation of the matrices (see \cite{NewChannelModel}), both approaches are equivalent.

It is therefore possible to use both, the scattering, as well as the impedance formulation to model the \ac{RIS} system under mutual coupling and one can also switch between the two descriptions when modeling the system.
When mutual coupling is included in the system model, both, the scattering and the impedance descriptions become intricate and
the system optimization and analysis gets more complicated.
For example, in \cite{MutualCouplingZAlgo}, mutual coupling has been analyzed based on impedance matrices and an algorithm using the Neumann series has been proposed for maximizing the \ac{SISO} channel gain.
The impact of mutual coupling on the rate has been further investigated in \cite{MutualCouplingSumRate} and \cite{MutualCouplingGradient}.
Additionally, in \cite{FollowUPSZJournal}, two mutual coupling aware algorithms based on the scattering parameters have been proposed which outperform the method in \cite{MutualCouplingZAlgo}.
All these algorithms rely on the Neumann series approximation whereas in \cite{ElementWise}, an element-wise algorithm has been derived which requires no approximation and provides closed-form solutions for each reactance.
In our previous work \cite{ConferenceDecoupling}, we propose an element-wise algorithm with reduced complexity.
All these solutions require an iterative algorithm to optimize the system under mutual coupling even in the simple case of maximizing a single-user \ac{SISO} channel gain.
Furthermore, as the solution with coupling requires an iterative algorithm it is difficult to analytically analyze the system.
Mutual coupling has also been investigated for the \ac{BD}-\ac{RIS} in \cite{FollowUpBDRIS} based on the Neumann series.
Similarly, the solution provided in this work also requires an iterative algorithm.

In our conference paper \cite{ConferenceDecoupling}, we propose decoupling networks as a solution to benefit from the mutual coupling at the \ac{RIS} array.
These networks have already been well investigated for the conventional transmit array (see, e.g., \cite{TowardCircuitTheory}).
In \cite{ConferenceDecoupling}, we demonstrate that the decoupling network transforms the system into a structure which is equivalent to a sytem without mutual coupling for which the coupling matrices are transferred into the channel matrices.
This simplifies the system model drastically as the model without coupling is significantly easier to handle and, additionally, well understood.
A direct consequence is that all methods (algorithms as well as theoretical analyses) without mutual coupling can be directly applied to the case of mutual coupling by incorporating decoupling networks into the model.
For example, the solution of a single-user \ac{SISO} channel gain maximization without mutual coupling can be given in closed-form \cite{MutualCouplingZAlgo}.
Hence, by using decoupling networks, a closed-form solution exists even in the case when considering mutual coupling (see \cite{ConferenceDecoupling}).
This has led to an analytical analysis of mutual coupling in \cite{ConferenceDecoupling} and we have shown that super-quadratic gains are possible under mutual coupling.

The resulting channel structure due to the decoupling networks which we have introduced in our conference paper \cite{ConferenceDecoupling} has recently been applied to analyze the BD-RIS with mutual coupling \cite{BDRISCoupling}.

In this article, we extend our results of \cite{ConferenceDecoupling} providing a thorough analysis of decoupling networks for the \ac{RIS} array.
In particular, we provide the following contributions:
\begin{itemize}
    \item Decoupling networks are proposed to handle the mutual coupling in the \ac{RIS} array.
          We show that when including these networks, the structure of the system model reduces to the case of no mutual coupling.
          Hence, all conventional solutions, neglecting the mutual coupling, can be directly extended to the case of mutual coupling when incorporating our proposed decouling networks.
          This comprises all algorithms as well as theoretical analyses.
    \item With decoupling networks, the equivalent adjustable impedance network obtains a non-diagonal structure even when considering a single-connected \ac{RIS}.
          While the equivalent matrix is non-diagonal, our method is significantly different to \ac{BD}-\acp{RIS} in which non-diagonal adjustable impedance networks are also used.
          Hence, we provide a thorough comparison of the approaches and highlight their differences.
    \item With decoupling networks it is possible to optimize the channel gain of a \ac{RIS}-aided \ac{SISO} system in closed form which allows to analytically analyze the array gain of a \ac{RIS} array under mutual coupling.
          From this analysis, we can draw connections to the conventional transmit array gain.
          In particular, we prove that a super-quadratic array gain which scales with $\NRe^4$ is possible in a \ac{RIS}-aided scenario. 
          In fact, this result provides novel insights in \ac{RIS}-aided communication since in previous work, it was assumed that only quadratic gains (i.e., $N^2$) are possible.
          Additionally, we compare the decoupling networks to the conventional \ac{RIS} (with and without mutual coupling) and show under which conditions the decoupling network is optimal.
          Furthermore, the important case of $\frac{\lambda}{2}$ antenna spacing is investigated.
    \item With the decoupling networks, we analyze both the single-connected \ac{RIS} as well as the \ac{BD}-\ac{RIS} under mutual coupling and show that regardless of the user position the \ac{BD}-\ac{RIS} channel gain scales at least cubic.
          Furthermore, we show that including losses in the model signficantly impacts the system's performance
    \item We demonstrate that the proposed decoupling networks result in superior performance w.r.t. the channel gain compared to other state-of-the-art methods.
\end{itemize}

\section{System Model}
We consider a \ac{RIS}-aided point-to-point \ac{MIMO} system where one \ac{BS}, having $\NB$ antennas, transmits data signals to $\NM$ receive antennas.
The transmission is supported by a \ac{RIS}, consisting of $\NRe$ reflecting elements. By using
the impedance representation, the channel from the \ac{BS} to the receiver reads as (see \cite{MutualCouplingAware}, \cite{NewChannelModel})
\begin{equation}\label{eq:ChannelModel}
    \begin{aligned}
    \bv_\teL &= \bD \bv_\tG = \frac{1}{4R} \bZ \bv_\tG\quad \text{with}\\
    \bZ &=  \bZ_\teDS - \bZ_\teDR( \bZ_\teR+ \bZ_\tN )^{\inv}\bZ_\teRS\\
    \end{aligned}
\end{equation}
with the definition of $\bv_\tG$ and $\bv_\tL$ according to Fig. \ref{fig:Three_Port}.
Here, $\bZ_\teDS$ is the direct channel from the \ac{BS} to the user, $\bZ_\teDR$ is the channel from the \ac{RIS} to the user, and $\bZ_\teRS$ the channel from the \ac{BS} to the \ac{RIS}.
The adjustable passive impedance network of the \ac{RIS} is given by $\bZ_\tN$,  whereas $\bZ_\teR$ is the impedance matrix of the \ac{RIS} antenna array accounting for mutual coupling.
If there were no mutual coupling present, $\bZ_\teR = \eyeM R$ would be a diagonal matrix.
In this article, we analyze the aspects of mutual coupling.
For isotropic radiators (which we consider in this article), the matrix $\bZ_\teR$ is defined as (see \cite[Eqn. (5)]{NewChannelModel}) 
\begin{equation}
    [\bZ_\teR]_{i,j} =    -\frac{R}{\im k d  \abs{i-j} }e^{-\im k  d \abs{i-j} } \quad \text{for} \; i\neq j,
\end{equation}
where $k= \frac{2 \pi}{\lambda}$ is the wave number and $d$ is the element spacing in wavelengths.
The diagonal entries are given by  $[\bZ_\teR]_{i,i} = R$.
Throughout this article, we assume the adjustable impedance network $\bZ_\tN$ to be lossless and reciprocal and, hence,
the matrix $\bZ_\tN$ 
\begin{equation}\label{eq:AdjNetwork}
    \bZ_\tN = \im\bm{X}_\tN
\end{equation}
with $\bm{X}_\tN=\bm{X}_\tN^{\transpo} \in \mathbb{R}^{N \times N}$ is real-valued. 

\section{Decoupling Networks}
To handle the mutual coupling between the reflecting elements at the \ac{RIS}, we propose decoupling networks.
These networks are well studied for the transmit and receive arrays (see, e.g., \cite{TowardCircuitTheory}) for handling the coupling between the antenna elements.
In addition to the transmit and receive arrays, we are now also assuming a decoupling network at the \ac{RIS} array according to Fig. \ref{fig:Three_Port}.
The performance of the decoupling networks for handling the mutual coupling in the \ac{RIS} array is analyzed in the following.

\subsection{Equivalent Adjustable Impedance Network}
In this section, we derive the new channel model when incorporating a decoupling network $\bZ_\teDN$ for the \ac{RIS} array according to Fig. \ref{fig:Three_Port}.
Please note that $\bZ_\teDN = \bZ_\teDN^\transpo$ and $\Real{ (\bZ_\teDN) }= \bzero$ holds as the decoupling network is assumed to be reciprocal and lossless.
Without the decoupling network, the recconfigurable network can be described by the equation
\begin{equation}\label{eq:ConvDecouplingNetwork}
    \bv_\teR = - \bZ_\tN \bi_\teR.
\end{equation}
Including now the decoupling network according to Fig. \ref{fig:Three_Port} leads to the new relationship 
\begin{equation}\label{eq:NewZN}
    \bv^\Dec_\teR = - \bZ_\tN^{\Dec} \bi^\Dec_\teR.
\end{equation}
In the following, we derive $\bZ_\tN^\Dec$.
With the decoupling network $\bZ_\teDN$, we obtain
\begin{equation}\label{eq:DecouplingNetworkGeneral}
    \begin{bmatrix}
        \bv_\teR\\
        \bv_\teR^\Dec
    \end{bmatrix}
    = 
    \begin{bmatrix}
         \bZ_{\teDN,11}& \bZ_{\teDN,12}\\
         \bZ_{\teDN,12}^\transpo& \bZ_{\teDN,22}
    \end{bmatrix}
\begin{bmatrix}
    \bi_\teR\\
    -\bi_\teR^\Dec
\end{bmatrix}
\end{equation}
with the four $\NRe \times \NRe$ matrix blocks $\bZ_{\teDN,ij}$. 
Substituting \eqref{eq:ConvDecouplingNetwork} into the first line of \eqref{eq:DecouplingNetworkGeneral}, we obtain 
\begin{equation}
    \bi_\teR =  (\bZ_{\teDN,11} +  \bZ_\tN)^{\inv}\bZ_{\teDN,12}\bi_\teR^\Dec.
\end{equation}
Together with the second line of \eqref{eq:DecouplingNetworkGeneral}, we arrive at
\begin{equation}
    \bv_\teR^\Dec = - (\bZ_{\teDN,22}-\bZ_{\teDN,12}^\transpo(\bZ_{\teDN,11} +  \bZ_\tN)^{\inv}\bZ_{\teDN,12})\bi_\teR^\Dec\hspace*{-6pt}.
\end{equation}
Therefore, we have found the new matrix [cf. \eqref{eq:NewZN}]
\begin{equation}\label{eq:NewAdjustImpNet}
    \bZ_\tN^\Dec = \bZ_{\teDN,22}-\bZ_{\teDN,12}^\transpo(\bZ_{\teDN,11} +  \bZ_\tN)^{\inv}\bZ_{\teDN,12}
\end{equation}
with the new channel model [cf. \eqref{eq:ChannelModel}]
\begin{equation}\label{eq:DecouplingChannelModel}
    \bZ^\Dec =  \bZ_\teDS - \bZ_\teDR( \bZ_\teR+  \bZ_\tN^\Dec)^{\inv}\bZ_\teRS.\\
\end{equation}
Comparing this channel model with the conventional one in \eqref{eq:ChannelModel},
we can see that the adjustable impedance network $\bZ_{\tN}$ is replaced by $\bZ_{\tN}^{\Dec}$ of \eqref{eq:NewAdjustImpNet} due to the decoupling network.
Therefore, a decoupling network only changes the adjustable impedance network by transforming the original impedance network $\bZ_\tN$  into $\bZ_\tN^{\Dec}$.
With a decoupling network $\bZ_{\teDN}$, we can now manipulate the adjustable impedance network $\bZ_\tN$ for the channel in \eqref{eq:DecouplingChannelModel} by choosing properly designed matrices in \eqref{eq:NewAdjustImpNet}.
Various choices exist for decoupling networks, see, e.g., \cite{TowardCircuitTheory}.
The only restriction is that the network should be lossless and reciprocal, i.e., $\bZ_\teDN = \bZ_\teDN^\transpo$ and $\Real{ (\bZ_\teDN) }= \bzero$.
However, there is one particular network structure which allows to completely decouple the \ac{RIS} array.
This is the power matching network (see \cite[(103)]{TowardCircuitTheory}) given by
\begin{equation}
    \bZ_{\teDN} = -\im \begin{bmatrix}
        \bzero& \sqrt{R} \Real{(\bZ_\teR)}^{\frac{1}{2}}\\
        \sqrt{R} \Real{(\bZ_\teR)}^{\frac{1}{2}}&  \Imag{(\bZ_\teR)}
    \end{bmatrix}
\end{equation}
with $ \Real{(\bZ_\teR)} = \Real{(\bZ_\teR)}^{\frac{1}{2}} \Real{(\bZ_\teR)}^{\frac{1}{2}}$ and we particularly focus on this choice throughout this article.
By using the power matching network, we arrive at the new adjustable impedance matrix 
\begin{equation}\label{eq:NewAdjustImpNetDecoupPowerMatch}
    \bZ_\tN^\Dec = -\im  \Imag{(\bZ_\teR)} + R \Real{(\bZ_\teR)}^{\frac{1}{2}} \bZ_\tN^{\inv} \Real{(\bZ_\teR)}^{\frac{1}{2}}.
\end{equation}
From \eqref{eq:NewAdjustImpNetDecoupPowerMatch} we can see that a lossless and reciprocal adjustable impedance network $\bZ_\tN$ directly results in a lossless and reciprocal $\bZ_\tN^\Dec$ as we have $\Real{(\bZ_\tN^\Dec)} = 0$ and $\bZ_\tN^\Dec = \bZ_\tN^{\Dec,\transpo}$.
This property results from the assumption that the decoupling network is also lossless and reciprocal. 
Furthermore, we can see that even if we assume a diagonal \ac{RIS} and, hence, the original impedance network $\bZ_\tN$ is diagonal,
 the new impedance matrix $\bZ_{\tN}^{\Dec}$ is in general non-diagonal.
A non-diagonal impedance network is typically referred to the so called \ac{BD}-\acp{RIS}.
Hence, the combination of an impedance network together with the decoupling network, resulting in a non-diagonal $\bZ^\Dec$, can be viewed as a particular type of \ac{BD}-\ac{RIS}.
However, significant differences exist when comparing them to the conventional \ac{BD}-\ac{RIS} architecture.
Please see Section \ref{sec:BDRISCompare} for details.

\subsection{Channel Model with Decoupling Networks}
In the following, we derive the new channel model when incorporating a decoupling network.
We start by rewriting the matrix $ \bZ_\teR + \bZ_\tN^\Dec$ within the inverse in \eqref{eq:DecouplingChannelModel}.
By inserting the definintion of the new adjustable impedance network accoring to \eqref{eq:NewAdjustImpNetDecoupPowerMatch} we arrive at
\begin{equation}\label{eq:InverseChannelModelDecoupMatching}
    \begin{aligned}
    \bZ_\teR + \bZ_\tN^\Dec &=\bZ_\teR - \im \Imag{(\bZ_\teR)}+R\Real{(\bZ_\teR)}^{\frac{1}{2}}\bZ_\tN^{\inv}\Real{(\bZ_\teR)}^{\frac{1}{2}}\\
    &= \Real{(\bZ_\teR)} + R \Real{(\bZ_\teR)}^{\frac{1}{2}} \bZ_\tN^{\inv} \Real{(\bZ_\teR)}^{\frac{1}{2}}\\
    &=\Real{(\bZ_\teR)}^{\frac{1}{2}} (R\eyeM + R^2\bZ_\tN^{\inv}) \frac{1}{R}\Real{(\bZ_\teR)}^{\frac{1}{2}}.\\
    \end{aligned}
\end{equation}
Hence, by defining the equivalent impedance network 
\begin{equation}
    \bm{\bar{Z}}_\tN = R^2\bZ_\tN^{\inv}
\end{equation}
as well as the effective impedance matrices as
\begin{equation}\label{eq:PowerMatchingNewChannelMatrices}
    \bm{\bar{Z}}_\teDR = \bZ_\teDR\Real{(\bZ_\teR)}^{-\frac{1}{2}}\sqrt{R}, \quad \bm{\bar{Z}}_\teRS = \sqrt{R}\Real{(\bZ_\teR)}^{-\frac{1}{2}}\bZ_\teRS,
\end{equation}
we arrive at the new channel model
\begin{equation}\label{eq:finalchannelmodeldecoupolingpowermatch}
    \begin{aligned}
        &\bZ^\Dec =  \bZ_\teDS - \bZ_\teDR( \bZ_\teR+  \bZ_\tN^\Dec)^{\inv}\bZ_\teRS\\
        &=  \bZ_\teDS - R\bZ_\teDR \Real{(\bZ_\teR)}^{-\frac{1}{2}} (\eyeM R +R^2\bZ_\tN^{\inv})^{\inv} \Real{(\bZ_\teR)}^{-\frac{1}{2}}\bZ_\teRS\\
        &=  \bZ_\teDS - \bm{\bar{Z}}_\teDR (\eyeM R +\bm{\bar{Z}}_\tN )^{\inv}\bm{\bar{Z}}_\teRS
    \end{aligned}
\end{equation}
which has the same structure as a \ac{RIS} channel model without mutual coupling.
Therefore, this is the same structure as the conventional channel model in \eqref{eq:ChannelModel} with the difference that instead of the non-diagonal $\bZ_\teR$ the scaled identity $\eyeM R$ is within the inverse.
The effect of mututal coupling is transformed into the new channel matrices $\bm{\bar{Z}}_\teDR$ and $\bm{\bar{Z}}_\teRS$.
This reduces the complexity of mutual coupling to the conventional system model without mutual coupling and all algorithms and methods can be extended to the case of mutual coupling by using these decoupling networks.

Before we focus on the analysis of this newly derived channel model, we give an alternative model based on the scattering parameters of the adjustable impedance model.
This alternative representaion with the scattering parameters is more popular in the literature and often leads to a representation which is easier to optimize.
Hence, we switch to the scattering description of the adjustable network $\bZ_\tN$ leading to
\begin{equation}
    \bZ_\tN = R (\eyeM + \bTheta)(\eyeM-\bTheta)^{\inv}.
\end{equation}
where $\bTheta$ is the equivalent scattering representation
\begin{equation}
    \ba_\teR = \bTheta \bb_\teR.
\end{equation}
The inverse is now given by 
\begin{equation}
    \bZ_\tN^{\inv} = \frac{1}{R}(\eyeM-\bTheta)(\eyeM + \bTheta)^{\inv}
\end{equation}
and we can write $(\eyeM R+R^2\bZ_\tN^{\inv})$ as 
\begin{equation}\label{eq:InverseScatteringFormulation}
    \begin{aligned}
        (\eyeM R +R^2\bZ_\tN^{\inv}) &=  \eyeM R + R (\eyeM-\bTheta)(\eyeM + \bTheta)^{\inv}\\
        &=R \left((\eyeM + \bTheta) + (\eyeM-\bTheta)\right)(\eyeM + \bTheta)^{\inv}\\
        &= 2R(\eyeM + \bTheta)^{\inv}.
    \end{aligned}
\end{equation}
Accordingly, by inserting \eqref{eq:InverseScatteringFormulation} and \eqref{eq:InverseChannelModelDecoupMatching} into \eqref{eq:DecouplingChannelModel},
 we arrive at the new channel model 
\begin{equation}\label{eq:finalchannelmodeldecoupolingpowermatch}
    \begin{aligned}
       & \bZ^\Dec =  \bZ_\teDS - \bZ_\teDR( \bZ_\teR+  \bZ_\tN^\Dec)^{\inv}\bZ_\teRS\\
        &=  \bZ_\teDS - R \bZ_\teDR \Real{(\bZ_\teR)}^{-\frac{1}{2}} (\eyeM R +R^2\bZ_\tN^{\inv})^{\inv} \Real{(\bZ_\teR)}^{-\frac{1}{2}}\bZ_\teRS\\
        &=  \bZ_\teDS - R \bm{\bar{Z}}_\teDR (\eyeM R +R^2\bZ_\tN^{\inv})^{\inv} \bm{\bar{Z}}_\teRS\\
        &=  \bZ_\teDS - \frac{1}{2R}\bm{\bar{Z}}_\teDR (\eyeM + \bTheta)\bm{\bar{Z}}_\teRS\\
        &= \bZ_\teDS + \frac{1}{2R}\bm{\bar{Z}}_\teDR(-\bTheta-\eyeM)\bm{\bar{Z}}_\teRS\\
        &= \bZ_\teDS + \frac{1}{2R}\bm{\bar{Z}}_\teDR(\bm{\bar{\Theta}}-\eyeM)\bm{\bar{Z}}_\teRS\\
    \end{aligned}
\end{equation}
with the equivalent scattering matrix 
\begin{equation}
    \bm{\bar{\Theta}} = - \bTheta.
\end{equation}
Summarizing the derivations in this section, we arrive at the two equivalent channel representations
\begin{equation}\label{eq:DecoupledChannelGainPowerMatching}
    \begin{aligned}
    \bZ^{\Dec}        &= \bZ_\teDS - \bm{\bar{Z}}_\teDR (\eyeM R +\bm{\bar{Z}}_\tN )^{\inv}\bm{\bar{Z}}_\teRS\\   
    &= \bZ_\teDS + \frac{1}{2R}\bm{\bar{Z}}_\teDR(\bm{\bar{\Theta}}-\eyeM)\bm{\bar{Z}}_\teRS\\     
\end{aligned}
\end{equation} 
which will be analyzed in the following.
For both representations in \eqref{eq:DecoupledChannelGainPowerMatching} we can notice that the channel model for decoupling networks shares the same structure as a system without mutual coupling between the \ac{RIS} elements. 
Hence, all algorithms and solutions, neglecting the mutual coupling of the \ac{RIS} array, can be extended to mutual coupling by considering the new channel matrices in \eqref{eq:PowerMatchingNewChannelMatrices}.
This is a significant advantage of decoupling networks and we discuss the performance of this approach in Section \ref{sec:ArrayPerfAnalysis}.
\subsection{Multiple $\frac{\lambda}{2}$ Spacing}\label{sec:LambdaHalf}
For the special case of a multiple $\frac{\lambda}{2}$ spacing, we have
\begin{equation}
    \bZ_\teR = \eyeM R, \quad \text{for} \quad d = \frac{\lambda}{2} k, \; k \in \mathbb{Z}^+,
\end{equation}
which results in  $\bm{\bar{Z}}_\teDR =  \bm{{Z}}_\teDR$ and $ \bm{\bar{Z}}_\teRS =  \bm{{Z}}_\teRS$ and, therefore in the channel model
\begin{equation}
    \begin{aligned}
    \bZ^{\Dec}        &= \bZ_\teDS - \bm{{Z}}_\teDR (\eyeM R +\bm{\bar{Z}}_\tN )^{\inv}\bm{{Z}}_\teRS.\\   
\end{aligned}
\end{equation} 
We can see that this is exactly the conventional channel model in \eqref{eq:ChannelModel} without mutual coupling.
Therefore, for a multiple $\frac{\lambda}{2}$ spacing the decoupled \ac{RIS} is equal to an uncoupled \ac{RIS}.
It is important that the conventional \ac{RIS} without decoupling networks is different to the uncoupled \ac{RIS} even for a multiple $\frac{\lambda}{2}$ spacing as the imaginary part of $\bZ_\teR$ is still included in the model which is non-zero for a multiple $\frac{\lambda}{2}$ spacing.
Only fully-connected \acp{RIS} are an exception as they lead to the same model with or without decoupling networks (see Section \ref{sec:FullyConnectedRISDecoupled} for details).

Furthermore, for this special antenna spacing, the implementation of the decoupled \ac{RIS} can be simplified.
The effective impedance network is given by
\begin{equation}
    \begin{aligned}
        \bZ_\tN^\Dec &= -\im  \Imag{(\bZ_\teR)} + R \Real{(\bZ_\teR)}^{\frac{1}{2}} \bZ_\tN^{\inv} \Real{(\bZ_\teR)}^{\frac{1}{2}}\\
        &= -\im  \Imag{(\bZ_\teR)} + R\bZ_\tN^{\inv}.
    \end{aligned}
\end{equation}
Therefore, it is possible to omit the decoupling networks and equivalently implement directly the impedance network $\bZ_\tN^\Dec $.
For example, in case of the single-connected \ac{RIS}, the adjustable impedance network is chosen as
\begin{equation}
    \bZ_\tN = -\im  \Imag{(\bZ_\teR)} -  \im \diag(\bx_\tN)
\end{equation}
where the diagonal elements are adjustable.
The off-diagonals are static and only depend on the mutual coupling which is the difference to a \ac{BD}-\ac{RIS} (see Section \ref{sec:BDRISCompare} for details).
\section{Connection to the BD-RIS}\label{sec:BDRISCompare}
We have already seen in \eqref{eq:NewAdjustImpNetDecoupPowerMatch} that independent of $\bZ_\tN$ being diagonal or non-diagonal, the new impedance network $\bZ_\tN^{\Dec}$ will be generally non-diagonal.
Hence, by including a decoupling network into the system model, the resulting network $\bZ_\tN^\Dec$ belongs to the \ac{BD}-\acp{RIS}.
However, we will see that this is a specific type of \ac{BD}-\ac{RIS} which is significantly different from conventional architectures.
In this section, we discuss the similarities and the differences of a single-connected \ac{RIS} with decoupling networks in comparison to a conventional \ac{BD}-\ac{RIS}.
\subsection{Diagonal RIS}
From \eqref{eq:NewAdjustImpNetDecoupPowerMatch} we know that the new adjustable impedance network with decoupling networks is given by 
\begin{equation}\label{eq:NewAdjustImpNetDecoupPowerMatchRecap}
    \bZ_\tN^\Dec = -\im  \Imag{(\bZ_\teR)} + R \Real{(\bZ_\teR)}^{\frac{1}{2}} \bZ_\tN^{\inv} \Real{(\bZ_\teR)}^{\frac{1}{2}}
\end{equation}
which generally results in a non-diagonal impedance matrix $\bZ_\tN^\Dec$ regardless of the structure of $\bZ_\tN$.
We now assume a diagonal impedance network 
\begin{equation}
    \bZ_\tN = \im \diag(\bx_\tN)
\end{equation}
for which the new adjustable impedance network $\bZ_\tN^\Dec$ is still non-diagonal.
However, the important difference to \ac{BD}-\acp{RIS} as they are discussed in the current literature is that the non-diagonal structure of the adjustable impedance network in \eqref{eq:NewAdjustImpNetDecoupPowerMatchRecap} only results from the mutual coupling.
As it only depends on the mutual coupling in the system, the decoupling network is static and independent of the scenario.
Hence, the elements of the decoupling networks are not reconfigurable and only have to be calibrated once at the fabrication of the array.
Afterwards, the decoupling networks does not change and is valid for any scenario.

This is the difference to the \ac{BD}-\acp{RIS} for which the non-diagonal structure is a reconfigurable design parameter that depends on the scenario.
In this case, the matrix $\bZ_\tN$ does not only depend on the mutual coupling of the antenna array but is also reconfigured based on the channel realizations or alternatively, its statistics.

It is, however, important to note that decoupling networks and \ac{BD}-\acp{RIS} are not competing architectures as they follow inherently different concepts.
The decoupling networks are specifically designed for handling the mutual coupling whereas \ac{BD}-\acp{RIS} have more degrees of freedom for optimizing the \ac{RIS} at the cost of an increased complexity.
This aspect is independent of whether there exists mutual coupling in the system or not.
Specifically, decoupling networks are only designed for the mutual coupling of the \ac{RIS} array and can be combined with any \ac{RIS} architecture, including both, the diagonal and \ac{BD} architectures. 

This is further illustrated by the new channel model for deocupling networks according to \eqref{eq:DecoupledChannelGainPowerMatching}
\begin{equation}
    \begin{aligned}
    \bZ^{\Dec}       
    &= \bZ_\teDS + \frac{1}{2R}\bm{\bar{Z}}_\teDR(\bm{\bar{\Theta}}-\eyeM)\bm{\bar{Z}}_\teRS.\\     
\end{aligned}
\end{equation} 
While the decoupling network has already been applied, the constraint set of $\bm{\bar{\Theta}}$ (or in the impedance representation $\bm{\bar{Z}}_\tN$) still plays a major role.
Therefore, even when using a decoupling network, the architecture of $\bm{\bar{\Theta}}$ ($\bm{\bar{Z}}_\tN$) being diagonal or non-diagonal matters.
For example, when combining the decoupling network with a diagonal \ac{RIS} we have the popular unit-modulus constraint set
\begin{equation}
    \bm{\bar{\Theta}} = \diag(\bm{\bar{\theta}}), \quad \abs{\bar{\theta}_n} =1\, \forall n.
\end{equation}
Whereas combining the decoupling networks with a fully-connected \ac{BD}-\ac{RIS} results in the constraint set
\begin{equation}
    \bm{\bar{\Theta}}= \bm{\bar{\Theta}}^{\transpo} \quad \text{and} \quad \bm{\bar{\Theta}}^{\He}\bm{\bar{\Theta}} = \eyeM.
\end{equation}
Additionally, decoupling networks can also be combined with other, partially connected, \ac{BD}-\ac{RIS} structures.
Please see, e.g., \cite{ScatteringRIS} for the different architectures.
However, in this article we will only analyze fully-connected and diagonal \acp{RIS}.

In summary, a decoupling network deals with the mutual coupling of the \ac{RIS} array and by that it transforms the system model with mutual coupling into a structure that resembles a system model without mutual coupling.
However, the decoupling networks can be combined with either a diagonal \ac{RIS} or a \ac{BD}-\ac{RIS} architecture.

\subsection{Fully-Connected BD-RIS}\label{sec:FullyConnectedRISDecoupled}
We will now focus on the fully-connected architecture where we can choose an arbitrary impedance matrix
\begin{equation}
    \bZ_\tN = \im \bm{X}_\tN
\end{equation}
with the only restriction that it is lossless and reciprocal.
As the fully-connected \ac{BD}-\ac{RIS} is already capable of implementing any non-diagonal (lossless and reciprocal) impedance matrix, 
the decoupling networks are not beneficial in this case. 
When combining the fully-connected tunable network $\bZ_\tN= \im \bm{X}_\tN$ with the decoupling network $\bZ_\teDN$, the effective network according to \eqref{eq:NewAdjustImpNetDecoupPowerMatch} is given by
\begin{equation}\label{eq:DecoupleBDRIS}
    \bZ_\tN^\Dec = -\im  \Imag{(\bZ_\teR)} -\im R \Real{(\bZ_\teR)}^{\frac{1}{2}}   \bm{X}_\tN^{\inv} \Real{(\bZ_\teR)}^{\frac{1}{2}}.
\end{equation}
This network is again lossless and reciprocal as $\Real(\bZ_\tN^\Dec) = \bm{0}$ and $\bZ_\tN^\Dec=\bZ_\tN^{\Dec,\transpo}$. 
As $\bZ_\tN^\Dec$ is a lossless and reciprocal network, there is no advantage of using a decoupling network, as the new network $\bZ_\tN^{\Dec}$ could be directly implemented with the \ac{BD}-\ac{RIS}
by choosing 
\begin{equation}\label{eq:DecoupleBDRISEquivalent}
    \bZ_\tN= -\im  \Imag{(\bZ_\teR)} + \im \frac{1}{R} \Real{(\bZ_\teR)}^{\frac{1}{2}}  \bm{X}_\tN^{\prime}\Real{(\bZ_\teR)}^{\frac{1}{2}}
\end{equation}
where $\bm{X}_\tN^{\prime}$ is any lossless and reciprocal impedance network.
Therefore, for a fully connected \ac{RIS} a decoupling network is not beneficial and there is a bidirectional mapping between $\bZ_\tN^\Dec$ and $ \bZ_\tN$ in \eqref{eq:DecoupleBDRIS} and 
between $\bZ_\tN$ and $\bm{X}_\tN^{\prime}$ in \eqref{eq:DecoupleBDRISEquivalent}. 

It is indeed quite intuitive that the decoupling network is unneccessary for a fully-connected \ac{BD}-\ac{RIS}.
Since we can anyway design an arbitrary lossless reciprocal network when incorporating a fully-connected \ac{BD}-\ac{RIS}, connecting a second lossless reciprocal network does not give us more design choices.
While the decoupling networks provide no benefits in this case, it can still be exploited that their special structure leads to a channel decomposition corresponding to having no mutual coupling.
Using \eqref{eq:DecoupleBDRISEquivalent}, the equivalent channel for the \ac{BD}-\ac{RIS} can be written as
\begin{equation}\label{eq:DecoupleBDRISEquivalentChannel}
    \bZ =  \bZ_\teDS - \bm{\bar{Z}}_\teDR^{}( \eyeM R + \im \bm{{X}}^{\prime}_\tN )^{\inv} \bm{\bar{Z}}_\teRS^{}
\end{equation}
which follows the exact same structure when having no mutual coupling and is a direct consequence when utilizing decoupling networks. 
Hence, when using a fully-connected \ac{BD}-{RIS}, the decompositions in \eqref{eq:DecoupleBDRIS}--\eqref{eq:DecoupleBDRISEquivalentChannel} introduced in \cite{ConferenceDecoupling} are also present when having no decoupling network \cite{BDRISCoupling}.

\section{Channel and Array Gain Maximization}\label{sec:ArrayPerfAnalysis}
We have already seen that decoupling networks result in a channel model which has the same structure as a system without mutual coupling.
In this section, we analyze the resulting performance of the decoupling networks in comparison to a conventional \ac{RIS} with mutual coupling but without the decoupling networks.
Moreover, in this section, we also provide an analytic discussion on the influence of mutual coupling on \ac{RIS}-aided systems.
Specifically we consider the channel and array gain of a single-user \ac{SISO} system.
Hence, the \ac{BS} is assumed to have one antenna serving a single user with also one antenna.
The \ac{RIS}, however, can have an arbitrary number of \ac{RIS} elements $N$.
The channel model for this scenario, when using decoupling networks, reads according to \eqref{eq:DecoupledChannelGainPowerMatching} as 
\begin{equation}\label{eq:SISOChannelGain}
    z^{\Dec} = z_\teDS + \frac{1}{2R}\bm{\bar{z}}_\teDR^{\transpo} (\bm{\bar{\Theta}} - \eyeM)\bm{\bar{z}}_\teRS
\end{equation}
where the matrix $\bm{\bar{\Theta}}$ can be either diagonal or non-diagonal depending on the \ac{RIS} architecture.
As \eqref{eq:SISOChannelGain} has the structure of a system without mutual coupling, the maximization of $\abs{z^{\Dec}}^2$ can be given in closed-form and in case of a diagonal \ac{RIS} $\bm{\bar{\Theta}} = \diag(\bm{\bar{\theta}})$, the optimal solution reads as \cite{MutualCouplingZAlgo}
\begin{equation}
    \bm{\bar{\theta}} = \exp\left(\im \arg\left(z_\teDS - \frac{1}{2R}\bm{\bar{z}}_\teDR^{\transpo} \bm{\bar{z}}_\teRS\right)\bm{1} -\im \arg(\bm{\bar{z}}_\teDR \odot \bm{\bar{z}}_\teRS)\right)
\end{equation}
with the corresponding channel gain given by
\begin{equation}\label{eq:PowerMatchingClosedForm}
    \abs{z^{\Dec}_{\text{D}}}^2 = \left(\abs{z_\teDS -\frac{1}{2R}\bm{\bar{z}}_\teDR^{\transpo}\bm{\bar{z}}_\teRS } + \frac{1}{2R}\summe{n=1}{\NRe}\abs{\bar{z}_{\teDR,n}^{}}\abs{ \bar{z}_{\teRS,n}}\right)^2\hspace*{-3pt}.
\end{equation}
For the fully-connected \ac{BD}-\ac{RIS}, on the other hand, we use an upper bound (see \cite{ScatteringRIS}) for which the channel gain is given by 
\begin{equation}\label{eq:BDUpperBoundChannel}
    \abs{z^{\Dec}_{\text{BD}}}^2 \le \left(\abs{\bar{z}_\teDS -\frac{1}{2R}\bar{z}_\teDR^{\transpo}\bar{z}_\teRS } + \frac{1}{2R} \norm{\bar{z}_{\teDR}^{}}_2\norm{ \bar{z}_{\teRS}}_2\right)^2.
\end{equation}
Interestingly, in \cite{BDRISOptimalSolution}, it has been shown that this upper bound can be achieved with the fully-connected \ac{BD}-\ac{RIS} and, hence, \eqref{eq:BDUpperBoundChannel} is the actual channel gain of the \ac{BD}-\ac{RIS}.
Throughout this section we assume the \ac{BD}-\ac{RIS} to be always fully-connected.
As the \ac{BD}-\ac{RIS} has more degrees of freedom we can already see that
\begin{equation}
    \abs{z^{\Dec}_{\text{D}}}^2 \le  \abs{z^{\Dec}_{\text{BD}}}^2
\end{equation}
holds and the \ac{BD}-\ac{RIS} leads, therefore, in general to a better performance.
\subsection{Scenario}
To analyze the channel/array gain for the different architectures,
we assume a \ac{LOS} scenario according to Fig. \ref{fig:Scenario} where the direct channel from the \ac{BS} to the user is blocked, i.e., $z_\teDS = 0 \, \Omega$.
\begin{figure}[h!]

	\centering
	\resizebox{0.3 \textwidth}{!}{
	\includegraphics*{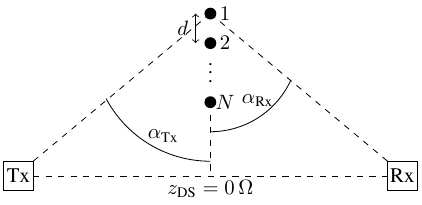}}

	\caption{\ac{SISO} Link with $\NRe$ Elements at the \ac{RIS}.}
	\label{fig:Scenario}
	
\end{figure}

We assume a \ac{ULA} at the \ac{RIS} which results in the channels
\begin{equation}
    z_\teDS = 0 \,\Omega, \quad \bz_\teDR^\transpo = \sqrt{\gamma_{\teDR}} \ba^\transpo_\teDR R, \quad  \bz_\teRS = \sqrt{\gamma_{\teRS}} \ba_\teRS R
\end{equation}
where $\sqrt{\gamma_{\teDR}}$, $\sqrt{\gamma_{\teRS}}$ are the pathlosses and $\ba_\teDR = \ba(\alpha_{\text{Rx}}), \; \ba_\teRS = \ba(\alpha_{\text{Tx}})$ the \ac{LOS} \ac{ULA} vectors where the $n$-th entry $a_n(\alpha)$ is given by
\begin{equation}
    a_n(\alpha) = e^{-\im(n-1) 2 \pi \frac{d}{\lambda} \cos(\alpha)}.
\end{equation}
The channel gain for one \ac{RIS} element ($N=1$), reads for both, the diagonal and the \ac{BD}-\ac{RIS} as $\abs{z^\Dec}^2=  \gamma_{\teDR}\gamma_{\teRS}R^2$.
Normalizing the channel gain for $\NRe$ elements by the gain of one element, we arrive at the array gain
\begin{equation}\label{eq:ArrayGainD}
    A_{\text{D}}^{\Dec} = \frac{1}{4}\hspace*{-2pt} \left(\abs{\ba_\teDR^{\transpo}\bC_\teR^{\inv}\ba_\teRS}+\summe{n=1}{\NRe}\abs{\ba_\teDR^\transpo\bC_\teR^{-\frac{1}{2}}\be_n}\hspace*{-2pt}\abs{\be^\transpo_n\bC_\teR^{-\frac{1}{2}}\ba_\teRS} \right)^{\hspace*{-4pt}2}\hspace*{-3pt}
\end{equation}
for the diagonal \ac{RIS} and at
\begin{equation}\label{eq:ArrayGainBD}
    \begin{aligned}
        A_{\text{BD}}^{\Dec} &= \frac{1}{4} \left(\abs{\ba_\teDR^{\transpo}\bC_\teR^{\inv}\ba_\teRS}+ \sqrt{\ba_\teDR^\Her\bC_\teR^{-1}\ba_\teDR \, \ba_\teRS^{\Her}\bC_\teR^{-1}\ba_\teRS} \right)^2
    \end{aligned}
\end{equation}
for the \ac{BD}-\ac{RIS} where we used the definition
\begin{equation}
    \bC_\teR = \frac{1}{R} \Real{(\bZ_\teR)}.
\end{equation}
It of course still holds that 
\begin{equation}
    A_{\text{D}}^{\Dec} \le A_{\text{BD}}^{\Dec}.
\end{equation}
Additionally, we know that the fully-connected \ac{BD}-\ac{RIS} without a decoupling network is equivalent to a fully-connected \ac{BD}-\ac{RIS} with a decoupling network.
Hence, we also have 
\begin{equation}\label{eq:BDSamewwoDecoup}
    A_{\text{BD}}^{\Dec} = A_{\text{BD}}
\end{equation}
where $A_{\text{BD}}$ is the array gain of a conventional system without the decoupling network.
Thus, $A_{\text{BD}}^{\Dec}$ and $A_{\text{BD}}$ can be used interchangeable.

First, we compare the \ac{BD}-\ac{RIS} and the diagonal \ac{RIS} structure for the setup $\alpha_{\text{Tx}}=0$, $\alpha_{\text{Rx}}=\frac{\pi}{2}$ in Fig. \ref{fig:Corner}.
\begin{figure}[!t]

	\centering
    \includegraphics*{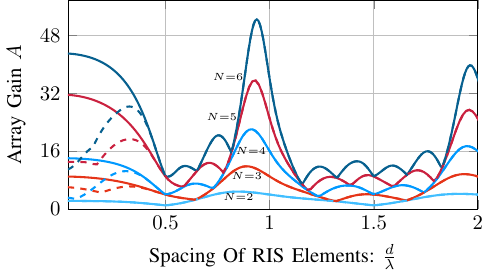}
    \vspace*{-0.3cm}
	\caption{Array Gain for $\alpha_{\text{Tx}}=0$ and $\alpha_{\text{Rx}}=\frac{\pi}{2}$. The BD-RIS has solid curves and the diagonal RIS has dashed curves.}
\label{fig:Corner}
\end{figure}
Interestingly, for small element spacings the diagonal \ac{RIS} (dashed curves) with $2k-1$ elements performs better performance than with $2k$ elements.
For the \ac{BD}-\ac{RIS} (solid curves), this is different and increasing the number of elements monotonically improves the performance for any spacing.
Additionally, we can observe that for small element spacings, the \ac{BD}-\ac{RIS} leads to a better performance in comparison to the diagonal \ac{RIS}, 
whereas for an increased spacing $d \ge \frac{\lambda}{2}$ there is no difference between the two architectures.
Hence, we will further study the transition point $d=\frac{\lambda}{2}$ in the following.

\subsection{Multiple $\frac{\lambda}{2}$ spacing}
In the case of having an element spacing of 
\begin{equation}
    d = \frac{\lambda}{2} k, \; k \in \mathbb{Z}^+,
\end{equation}
the matrix  $\bC_\teR$ results in the identity matrix, i.e., 
\begin{equation}
    \bC_\teR = \eyeM.
\end{equation}
From Section \ref{sec:LambdaHalf} we know that in this case a decoupled \ac{RIS} is equal to an uncoupled \ac{RIS} whereas the conventional \ac{RIS} without decoupling networks leads to a different model.
Only the fully-connected \ac{BD}-\ac{RIS} is an exception as the system model is identical with or without decoupling networks (see Section \ref{sec:FullyConnectedRISDecoupled}).
We are comparing the performance of the various architectures in case of a multiple $\frac{\lambda}{2}$ spacing.
The channel gain of a diagonal \ac{RIS} with decoupling networks is given as
\begin{equation}
    \begin{aligned}
        A_{\text{D}}^{\Dec} &= \frac{1}{4} \left(\abs{\ba_\teDR^{\transpo}\ba_\teRS}+\summe{n=1}{\NRe}\abs{\ba_\teDR^\transpo\be_n}\abs{\be^\transpo_n\ba_\teRS} \right)^2\\
        & = \frac{1}{4} \left(\abs{\ba_\teDR^{\transpo}\ba_\teRS}+N \right)^2.
    \end{aligned}
\end{equation}
Additionally, for the \ac{BD}-\ac{RIS} we arrive at the channel gain
\begin{equation}
    \begin{aligned}
        A_{\text{BD}} = A_{\text{BD}}^{\Dec} &= \frac{1}{4} \left(\abs{\ba_\teDR^{\transpo}\ba_\teRS}+ \sqrt{\ba_\teDR^\Her\ba_\teDR  \ba_\teRS^{\Her}\ba_\teRS} \right)^2\\
        & = \frac{1}{4} \left(\abs{\ba_\teDR^{\transpo}\ba_\teRS}+N \right)^2\\
        &= A_{\text{D}}^{\Dec}
    \end{aligned}
\end{equation}
where $A_{\text{BD}} = A_{\text{BD}}^{\Dec}$ is given in \eqref{eq:BDSamewwoDecoup}. 
Therefore, the \ac{BD}-\ac{RIS} and the diagonal \ac{RIS} have the same performance 
\begin{equation}
    A_{\text{BD}} = A_{\text{BD}}^{\Dec} =  A_{\text{D}}^{\Dec}
\end{equation}
for a  multiple $\frac{\lambda}{2}$ spacing and are equal to an uncoupled \ac{RIS} array.
With these comparisons it is also possible to compare the conventional diagonal \ac{RIS} with the decoupled \ac{RIS}.
The decoupled \ac{RIS} leads for this spacing always to a better performance as 
\begin{equation}
    A_{\text{D}}^{\Dec} = A_{\text{BD}} \ge A_{\text{D}}
\end{equation}
where $ A_{\text{D}}$ is the array gain of the conventional diagonal \ac{RIS}.
This also means that at $d=\frac{\lambda}{2}$ un uncoupled diagonal \ac{RIS} will always perform better in comparison to a diagonal \ac{RIS} with mutual coupling.

\begin{figure}[!t]
	\centering
    \includegraphics*{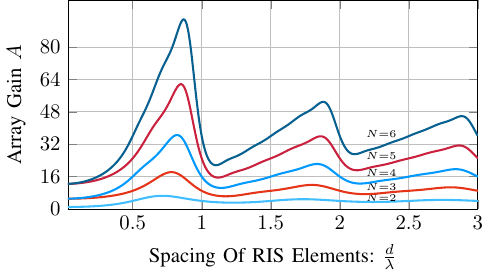}
    \vspace*{-0.3cm}
	\caption{Array Gain for the front-fire direction of the array.}
\label{fig:FrontFire}

\end{figure}

\subsection{Specular Reflection Scenario}
We now return to an arbitrary spacing $d$.
In particular, the specular reflection scenario is now investigated, where we have 
\begin{equation}
    \alpha_{\text{Tx}} = \pi - \alpha_{\text{Rx}}, \quad \text{and, hence, }\quad \ba_\teDR = \ba_\teRS^*.
\end{equation}
Two particular important scenarios, the front-fire and the end-fire direction belong to this specular type and are investigated below in detail.
For the specular case, the array gain of the diagonal \ac{RIS} is given as
\begin{equation}
    \begin{aligned}
        A_{\text{D}}^{\Dec} \hspace*{-3pt} &= \frac{1}{4} \left(\abs{\ba_\teDR^{\transpo}\bC_\teR^{\inv}\ba_\teRS}+\summe{n=1}{\NRe}\abs{\ba_\teDR^\transpo\bC_\teR^{-\frac{1}{2}}\be_n}\abs{\be^\transpo_n\bC_\teR^{-\frac{1}{2}}\ba_\teRS} \right)^{\hspace*{-3pt} 2}\\
        &= \frac{1}{4} \left(\ba_\teRS^{\Her}\bC_\teR^{\inv}\ba_\teRS+\summe{n=1}{\NRe}\abs{\ba_\teRS^\Her\bC_\teR^{-\frac{1}{2}}\be_n}\abs{\be^\transpo_n\bC_\teR^{-\frac{1}{2}}\ba_\teRS} \right)^2\\
        &= \frac{1}{4} \left(\ba_\teRS^{\Her}\bC_\teR^{\inv}\ba_\teRS+\norm{\bC_\teR^{-\frac{1}{2}}\ba_\teRS}^2 \right)^2\\
        &=  \left(\ba_\teRS^{\Her}\bC_\teR^{\inv}\ba_\teRS\right)^2\\
    \end{aligned}
\end{equation}
whereas the array gain of the \ac{BD}-\ac{RIS} is given by
\begin{equation}
    \begin{aligned}
       A_{\text{BD}} &= \frac{1}{4} \left(\abs{\ba_\teDR^{\transpo}\bC_\teR^{\inv}\ba_\teRS}+ \sqrt{\ba_\teDR^\Her\bC_\teR^{-1}\ba_\teDR \, \ba_\teRS^{\Her}\bC_\teR^{-1}\ba_\teRS} \right)^2\\
        &= \frac{1}{4} \left(\ba_\teRS^{\Her}\bC_\teR^{\inv}\ba_\teRS+ \ba_\teRS^{\Her}\bC_\teR^{-1}\ba_\teRS \right)^2\\
        &=  \left(\ba_\teRS^{\Her}\bC_\teR^{\inv}\ba_\teRS\right)^2.
    \end{aligned}
\end{equation}
Hence, the \ac{BD}-\ac{RIS} and the diagonal \ac{RIS} lead to the same performance
\begin{equation}\label{eq:SpecularSamePerf}
    A_{\text{D}}^{\Dec} = A_{\text{BD}} = \left(\ba_\teRS^{\Her}\bC_\teR^{\inv}\ba_\teRS\right)^2.
\end{equation} 
Interestingly, this is exactly the square of a conventional transmit array gain at the \ac{BS} (see \cite{TowardCircuitTheory} for the transmit array gain).
From \eqref{eq:SpecularSamePerf} it also directly follows that 
\begin{equation}
    A_{\text{D}}^{\Dec} = A_{\text{BD}} \ge A_{\text{D}}
\end{equation}
and, hence, the diagonal \ac{RIS} with decoupling network always outperforms the conventional diagonal \ac{RIS} in this specular reflection scenario.
Two particular scenarios are now further investigated, the front-fire as well as the end-fire directions.
\subsubsection*{Front-Fire}
In the case of having impinging waves in the front-fire direction we have $\alpha_{\text{Rx}} = \alpha_{\text{Tx}} =\frac{\pi}{2}$ and, hence, $\ba_\teRS = \ba_\teDR = \ones$.
The array gain, therefore, reads as
\begin{equation}
    \begin{aligned}\label{eq:FrontFireArrayGain}
        A_{\text{D}}^{\Dec} =  A_{\text{BD}} =  (\ones^\transpo\bC_\teR^{\inv}\ones)^2
    \end{aligned}
\end{equation}
and we arrive at Fig. \ref{fig:FrontFire}.
In fact, the array gains in Fig. \ref{fig:FrontFire} correspond to the ones in \cite[Fig. 6]{TowardCircuitTheory} with the key difference of having squared values.
We also notice that for $2k$ and for $2k-1$ elements the array gain converges to the same value for $ d \rightarrow 0$.
Please note, that the \ac{BD}-\ac{RIS} and the diagonal \ac{RIS} lead to the same performance in this scenario.
\subsubsection*{End-Fire}
In case of the end-fire direction we have $\alpha_{\text{Rx}}  = \pi,\; \alpha_{\text{Tx}} = 0$ and, hence, $\ba_\teRS = \ba_0,\;\ba_\teDR = \ba_0^*$ with $\ba_0 = \ba(0)$.
Therefore, we arrive at the array gain
\begin{equation}
        A_{\text{D}}^{\Dec} =  A_{\text{BD}}= \left(\ba_0^\Her \bC_\teR^{\inv}\ba_0 \right)^2 
\end{equation}
which is illustrated in Fig. \ref{fig:EndFire}.
\begin{figure}[!t]
	\centering
    \includegraphics*{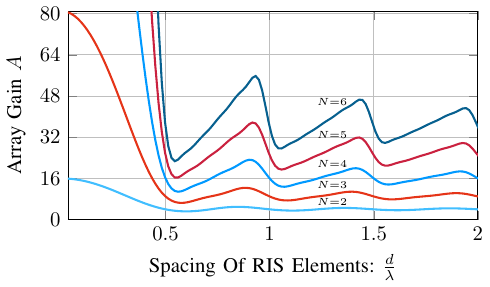}
    \vspace*{-0.3cm}
	\caption{Array Gain for the end-fire direction of the array.}
\label{fig:EndFire}
\end{figure}
Here, the gains correspond to those in \cite[Fig. 5]{TowardCircuitTheory} with squared values and we also observe a significant gain for $d\rightarrow 0$ as in \cite{TowardCircuitTheory}. 
By utilizing \cite{SquaredArrayGain} 
\begin{equation}\label{eq:QuadraticGainConventionalArray}
    \underset{d \rightarrow 0}{\lim} \ba_0^\Her \bC_\teR^{\inv}\ba_0 = N^2
\end{equation}
we obtain a quartic array gain in case of having a RIS in end-fire direction, i.e.,
\begin{equation}
    \underset{d\rightarrow 0}{\lim}\,A_{\text{D}}^{\Dec}= \underset{d\rightarrow 0}{\lim}\,A_{\text{BD}} = \underset{d\rightarrow 0}{\lim}\, (\ba_0^\Her \bC_\teR^{\inv}\ba_0)^2 = N^4.
\end{equation}
Hence, we can directly conclude that a super-quadratic array gain is possible when implementing a RIS in the end-fire direction.
Note again that in this case the \ac{BD}-\ac{RIS} and the diagonal \ac{RIS} lead to the same performance according to \eqref{eq:SpecularSamePerf}.
This super-quadratic gain clearly exceeds the gain of $N^2$ which has been generally assumed in existing literature.
The question arises whether this super-quadratic gain can be generalized to more cases except of being in end-fire direction and we will analyze this in the following.

\section{Super-Quadratic Gain}
We have seen that, specifically for the end-fire direction, a super-quadratic gain of
\begin{equation}\label{eq:endfirequartic}
    \underset{d\rightarrow 0}{\lim}\, (\ba_0^\Her \bC_\teR^{\inv}\ba_0)^2 = N^4
\end{equation}
is possible for both the diagonal \ac{RIS} as well as the \ac{BD}-\ac{RIS}.
In this section, we further investigate and generalize our result in \eqref{eq:endfirequartic}.
\subsection{Maximum at End-Fire}
Firstly, it is important to note that the end-fire direction gives the maximum possible gain in a \ac{RIS}-assisted scenario for $d\rightarrow 0$ which we summarize in the following theorem.
\begin{theorem}\label{theo:MaxGain}
    The maximum array gain of a decoupled \ac{RIS} for $d \rightarrow 0$ is given by 
    \begin{equation}
        \begin{aligned}
            \underset{d\rightarrow 0}{\lim}\,A_{\mathrm{D/BD}}^{\mathrm{DN}} &\le (\underset{d\rightarrow 0}{\lim}\,\ba_0^{\Her}\bC_\teR^{\inv}\ba_0)^2 = N^4
        \end{aligned}
    \end{equation}
which is achievable for both the \ac{BD} as well as the diagonal \ac{RIS} with decoupling networks.
Additionally, it follows that 
\begin{equation}
    A_{\mathrm{D}}^{\mathrm{DN}} = A_{\mathrm{BD}} \ge A_{\mathrm{D}}
\end{equation}
and, hence, the conventional \ac{RIS} without decoupling networks cannot exceed this gain.
\end{theorem}
\begin{proof}
    We start with the array gain of the \ac{BD}-\ac{RIS} \eqref{eq:ArrayGainBD}

\begin{equation}
    \begin{aligned}
      \abs{\ba_\teDR^{\transpo}\bC_\teR^{\inv}\ba_\teRS} &= \abs{\ba_\teDR^{\transpo}\bC_\teR^{-\frac{1}{2}}\bC_\teR^{-\frac{1}{2}}\ba_\teRS}\\
    &\le \norm{\ba_\teDR^{\transpo}\bC_\teR^{-\frac{1}{2}}}_2 \norm{\bC_\teR^{-\frac{1}{2}}\ba_\teRS}_2\\
    &=  \sqrt{\ba_\teDR^\Her\bC_\teR^{-1}\ba_\teDR \, \ba_\teRS^{\Her}\bC_\teR^{-1}\ba_\teRS}\\
    \end{aligned}
\end{equation}
and, hence, we can bound the gain of the \ac{BD}-\ac{RIS} as
\begin{equation}
    \begin{aligned}
        A_{\text{BD}} &= \frac{1}{4} \left(\abs{\ba_\teDR^{\transpo}\bC_\teR^{\inv}\ba_\teRS}+ \sqrt{\ba_\teDR^\Her\bC_\teR^{-1}\ba_\teDR \, \ba_\teRS^{\Her}\bC_\teR^{-1}\ba_\teRS} \right)^2\\
        &\le \ba_\teDR^\Her\bC_\teR^{-1}\ba_\teDR \, \ba_\teRS^{\Her}\bC_\teR^{-1}\ba_\teRS.
    \end{aligned}
\end{equation}
From the analysis of the transmit array \cite{SquaredArrayGain} we know that
\begin{equation}
    \underset{d\rightarrow 0}{\lim}\,\ba_\teRS^{\Her}\bC_\teR^{\inv}\ba_\teRS \le  \underset{d\rightarrow 0}{\lim}\,\ba_0^{\Her}\bC_\teR^{\inv}\ba_0 = N^2
\end{equation}
and, hence, the end-fire direction gives the largest gain for $d\rightarrow 0$.
Therefore, we arrive at 
\begin{equation}
    \underset{d\rightarrow 0}{\lim}\, A_{\text{BD}}  \le  (\underset{d\rightarrow 0}{\lim}\,\ba_0^{\Her}\bC_\teR^{\inv}\ba_0)^2 = N^4.
\end{equation}
This upper bound for $d \rightarrow 0$ is achievable in the end-fire direction.
It is important to note that this maximum gain is also achievable by the diagonal \ac{RIS} in the end-fire direction
\begin{equation}
    \begin{aligned}
        \underset{d\rightarrow 0}{\lim}\,A_{\text{D}}^{\Dec} &=(\underset{d\rightarrow 0}{\lim}\,\ba_0^{\Her}\bC_\teR^{\inv}\ba_0)^2 = N^4
    \end{aligned}
\end{equation}
and, hence, the \ac{BD}-\ac{RIS} and the diagonal \ac{RIS} have the same maximum channel gain which is quartic.
\end{proof}
\subsection{Super-Quadratic Gain of Coupled RISs}
From the last sections we know that the gain of an uncoupled \ac{RIS}/decoupled \ac{RIS} at $\frac{\lambda}{2}$ is given by 
\begin{equation}\label{eq:DecoupledUpperBound}
    \begin{aligned}
        A_{\text{BD}} = A_{\text{D}}^{\Dec} = \frac{1}{4} \left(\abs{\ba_\teDR^{\transpo}\ba_\teRS}+N \right)^2 \le N^2
    \end{aligned}
\end{equation}
whereas for the coupled \ac{RIS} for $d \rightarrow 0$ it is given by 
\begin{equation}
    \begin{aligned}
        A_{\text{BD}} = A_{\text{D}}^{\Dec} = N^4.
    \end{aligned}
\end{equation}
This shows that mutual coupling can lead to a significantly better performance in comparison to an uncoupled \ac{RIS} array (decoupled \ac{RIS} array with $d= \frac{\lambda}{2}$).
However, the requirement of using \eqref{eq:QuadraticGainConventionalArray} for deriving Theorem \ref{theo:MaxGain} indicates that the quartic gain is only achievable in the case of placing the BS, the RIS as well as the user exactly in end-fire direction.
Indeed, similar to the case for transmit arrays in \cite{TowardCircuitTheory}, when placing them in front-fire direction, mutual coupling for $d \rightarrow 0$ decreases the gain compared to having no mutual coupling (cf. \mbox{Fig \ref{fig:FrontFire}}).
However, there is a major difference in comparison to a conventional transmit array.
While we cannot know the position of the user and, hence, cannot ensure to be in end-fire direction, 
it is possible to position the \ac{RIS} and the \ac{BS} as desired.

Hence, it is possible to place the \ac{RIS} and the \ac{BS} in end-fire direction illustrated in Fig. \ref{fig:ScenarioCorner}.
\begin{figure}[h!]

	\centering
	
	\resizebox{0.3 \textwidth}{!}{
    \includegraphics*{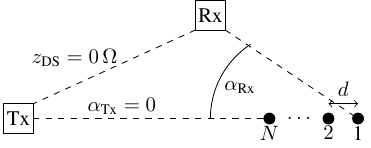}}
	\caption{BS and RIS positioned in end-fire direction.}
	\label{fig:ScenarioCorner}
	
\end{figure}
\begin{figure}[t!]
    \subfigure[$N=4$]{
    \includegraphics*{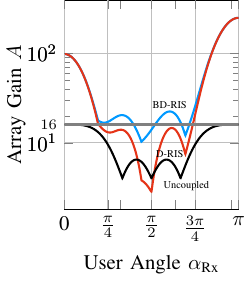}
    }\hspace*{0pt}\subfigure[$N=8$]
    {
        \includegraphics*{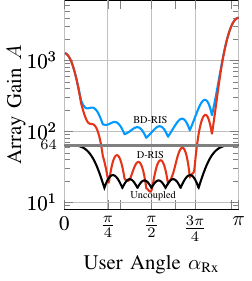}
    }
    \caption{Array Gain over the user angle for $d = \frac{\lambda}{32}$.}
    \label{fig:IdealTheorem}
\vspace*{-0.5cm}
\end{figure}
We already know that if the user is in the end-fire direction, i.e., $\alpha_{\text{Rx}}=\pi$, we observe a super-quadratic gain of $N^4$ for $d\rightarrow 0$, significantly exceeding the $N^2$ gain of an uncoupled \ac{RIS} (or decoupled \ac{RIS} with $\frac{\lambda}{2}$ spacing).
For other user angles $\alpha_{\text{Rx}}$ this is not the case.
However, we will show that it is still possible for the \ac{BD}-\ac{RIS} to achieve a super-quadratic gain in the order of $N^3$.
Since we position the BS and the \ac{RIS} in end-fire direction, we obtain the array gain (cf. \eqref{eq:ArrayGainBD})
\begin{equation}
    \begin{aligned}
       A_{\text{BD}} &= \frac{1}{4} \left(\abs{\ba_\teDR^{\transpo}\bC_\teR^{\inv}\ba_0}+ \sqrt{\ba_\teDR^\Her\bC_\teR^{-1}\ba_\teDR \, \ba_0^{\Her}\bC_\teR^{-1}\ba_0} \right)^2\\
       &\ge \frac{1}{4} \ba_\teDR^\Her\bC_\teR^{-1}\ba_\teDR \, \ba_0^{\Her}\bC_\teR^{-1}\ba_0.
    \end{aligned}
\end{equation}
Hence, for $d \rightarrow 0$, we can again utilize \eqref{eq:QuadraticGainConventionalArray} and observe that the array gain still exhibits a lower bound, i.e.,
\begin{equation}\label{eq:LowerBoundBDGeneral}
    \underset{d\rightarrow 0}{\lim}\,A_{\text{BD}} \ge  \frac{N^2}{4} \underset{d\rightarrow 0}{\lim}\,\ba_\teDR^\Her\bC_\teR^{-1}\ba_\teDR.
\end{equation}
In consequence, placing the BS and the RIS in end-fire direction leads to an $N^2$ scaling for small antenna spacings.
However, the lower bound in \eqref{eq:LowerBoundBDGeneral} also contains the term $\ba_\teDR^\Her\bC_\teR^{-1}\ba_\teDR$ which is the expression of a conventional transmit array gain and depends on the user angle $\alpha_{\text{Rx}}$.
This expression is analyzed in the following.
From \cite{SquaredArrayGain} we know that the conventional array gain for a coupled array reads as 
\begin{equation}\label{eq:AntennaGainChristoffel}
    \underset{d\rightarrow 0}{\lim}\,\ba_\teDR^\Her\bC_\teR^{-1}\ba_\teDR = \summe{n=0}{N-1}(2n+1)P^2_n(\cos(\alpha_{\text{Rx}}))
\end{equation}
where $P_n(x)$ is the Legendre polynomial of degree $n$.
To end up with a super-quadratic lower bound on the array gain in \eqref{eq:LowerBoundBDGeneral}, we start with the following proposition

\begin{proposition}\label{prop:GlobalMinimum}
    The function \mbox{$f_N(x) =  \summe{n=0}{N-1}(2n+1)P^2_n(x)$} with $x \in [-1,1]$ is lower bounded by 
    \begin{equation}
        f_N(x) \ge \frac{N}{2}\quad \forall x
    \end{equation}
    and its global minimum is given by $f(x_{\text{min}}) = N^2 P_N^2(x_{\text{min}})$. 
    For even $N$ the global minimum is attained at $x_{\text{min}}=0$, whereas for odd $N$
    the global minima are at $x_{\text{min}} = \pm x_0$ where $x_0$ is the positive zero of the derivative $P_N^\prime(x)$ which is closest to $x=0$. 
    \begin{proof}
        Please see Appendix \ref{app:GlobMin}.
    \end{proof}
\end{proposition}
\begin{corollary}\label{col:ArrayLowerBound}
    The conventional transmit antenna gain $\ba_{\mathrm{DR}}^\Her\bC_\teR^{-1}\ba_{\mathrm{DR}}$ for $ d \rightarrow 0$  attains its minimum gain in the front-fire direction ($\alpha_{\text{Rx}} = \frac{\pi}{2}$) for even $N$.
    Additionally, the antenna gain is bounded by
    \begin{equation}
        \underset{d\rightarrow 0}{\lim}\,\ba_{\mathrm{DR}}^\Her\bC_\teR^{-1}\ba_{\mathrm{DR}} \ge \frac{N}{2}.
        \vspace*{-0.cm}
    \end{equation}
    \begin{proof}
        This follows directly from Proposition \ref{prop:GlobalMinimum} and \eqref{eq:AntennaGainChristoffel}.
    \end{proof}
\end{corollary}
It should be noted that Corollary \ref{col:ArrayLowerBound} also holds for the conventional transmit array gain when no RIS is employed in the system.
By combining Corollary \ref{col:ArrayLowerBound} and \eqref{eq:LowerBoundBDGeneral}, we derive the following theorem.
\begin{theorem}\label{theo:LowerBoundRISGain}
    The array gain of a \ac{BD}-\ac{RIS} for $d \rightarrow 0$ achieves a super-quadratic gain of at least 
    \begin{equation}
        \underset{d\rightarrow 0}{\lim}\, A_{\text{BD}} \ge \frac{N^3}{8}.
        \vspace*{-0.2cm}
    \end{equation}
    This holds for any user angle $\alpha_{\text{Rx}}$.
\end{theorem}
By comparing the upper bound on the array gain of an uncoupled RIS in \eqref{eq:DecoupledUpperBound} with the lower bound on the array gain of BD-RISs in Theorem \ref{theo:LowerBoundRISGain}, we observe that for $N \geq 8$ the coupled RIS outperforms the uncoupled one when $d \rightarrow 0$ and considering no other losses.
More importantly, the coupled \ac{RIS} scales super-quadratic with the number of RIS elements as $N^3$ in comparison to an uncoupled \ac{RIS} which only scales with $N^2$.
Theorem \ref{theo:LowerBoundRISGain} is illustrated in Fig. \ref{fig:IdealTheorem}.
We can see that for both $N=4$ as well as for $N=8$ the \ac{BD}-\ac{RIS} leads to a better performance in comparison to the diagonal \ac{RIS}.
Only for the end-fire direction they achieve the same performance which validates our theoretical result in Theorem \ref{theo:MaxGain}.
Moreover, Fig. 7 also validates Theorem \ref{theo:LowerBoundRISGain}. 
For $N=4$, the \ac{BD}-\ac{RIS} achieves a good performance but it does not exceed $N^2$ for all user angles which is the maximum gain of the uncoupled \ac{RIS}.
This is different to $N=8$, where we can see that regardless of the user, the \ac{BD}-\ac{RIS} leads to a better performance.

\begin{figure}[t!]
    \includegraphics*{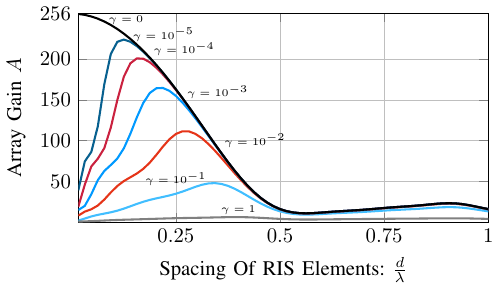}
\caption{Array Gain for the end-fire direction of the array for $N=4$.}
\label{fig:LossyEnd}
\end{figure}
\subsection{Ohmic Losses} 
We have seen in the last section that the coupled \ac{RIS} leads to significant improvements in comparison to an uncoupled \ac{RIS}, achieving an $N^4$ gain in the end-fire direction.
Moreover, the \ac{BD}-\ac{RIS} leads to at least a cubic gain for all user angles.
However, we considered lossless antenna arrays where the spacing $d \rightarrow 0$ and the simulations are based on a very small spacing of $d=\frac{\lambda}{32}$.
We will now investigate these super-quadratic gains under more realistic circumstances.
Equivalent to \cite{TowardCircuitTheory} we assume Ohmic losses and when including the dissipation resistance of $R_\text{d}$, we obtain the new coupling matrix 
\begin{equation}
    \bZ_\teR^{\text{loss}} = \bZ_\teR + \eyeM R_\text{d}.
\end{equation}
The real part of $\bZ_\teR$ with the values $R$ on the diagonal is now combined with the resistance $R_\text{d}$.
Similarly, we obtain
\begin{equation}
    \bC_\teR^{\text{loss}} =  \bZ_\teR^{\text{loss}}\frac{1}{R} = \bC_\teR + \gamma \eyeM
\end{equation}
with $\gamma = \frac{R_\text{d}}{R}$ and $R_\text{d}$ being the dissipation resistance. 
Hence, the ideal coupling matrix $\bC_\teR$ is now regularized by $\gamma \eyeM$.
We first study the end-fire direction where the largest (the $N^4$ gain) appears for the ideal lossless array.
Inlcuding Ohmic losses we arrive at the new expression
\begin{equation}
   A_{\text{D}}^{\Dec}=  A_{\text{BD}} = (\ba_0^\Her (\bC_\teR + \gamma \eyeM)^{\inv}\ba_0)^2
\end{equation}
for the end-fire direction,
where we can observe that the array gain decreases with an incresing $\gamma$.
Additionally, the higher the value of $\gamma$, the stronger the diagonal loading of the matrix $(\bC_\teR + \gamma \eyeM)$ and, hence, the more similar the scenario to the case without mutual coupling.

The resulting performance when considering Ohmic losses is illustrated in Fig. \ref{fig:LossyEnd} and resembles \cite[Fig. 11]{TowardCircuitTheory} with squared values.
\begin{figure}
    \includegraphics*{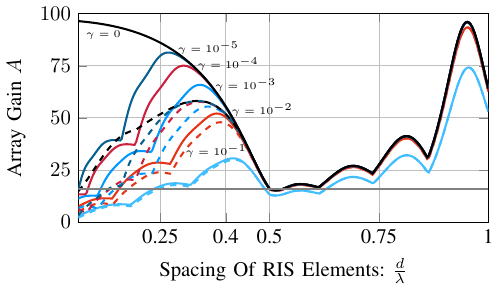}
	\caption{Array Gain for the $\alpha_{\text{Tx}}=0$ and $\alpha_{\text{Rx}}=\frac{\pi}{2}$ and $N=8$. The BD-RIS has solid curves and the diagonal RIS has dashed curves.}
\label{fig:LossyCorner}
\end{figure}

\begin{figure}[t!]
    \subfigure[Lossless $\gamma=0$]{
    \includegraphics*{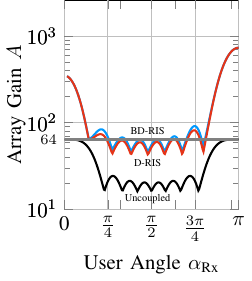}
    }\hspace*{0pt}\subfigure[Lossy ${\gamma=10}^{-2}$]
    {
        \includegraphics*{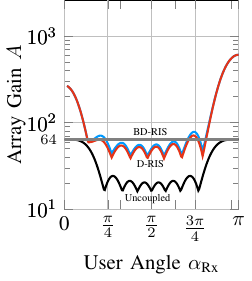}
    }
    \caption{Array Gain over the user angle for $d = 0.4$ and $N=8$.}
    \label{fig:AnglePlotRealistic}
    \end{figure}
We observe that a very small element spacing actually leads to a large degradation in the performance when considering Ohmic losses in the array.
Furthermore, the optimal spacing of the elements increases for increasing losses.
This behavior can also be observed for  $\alpha_{\text{Tx}}=0$ and $\alpha_{\text{Rx}}=\frac{\pi}{2}$ in Fig. \ref{fig:LossyCorner} where an increased spacing is necessary when losses are considered in the model.
The lossless curves in Fig. \ref{fig:LossyCorner} (see also Fig. \ref{fig:Corner}) show a significant difference between the \ac{BD}-\ac{RIS} and the diagonal \ac{RIS} for small spacings.
However, as the spacing increases, the difference between the two architectures decreases and eventually vanishes.
When inlcuding losses, an increased spacing is necessary and we can see that the difference between the \ac{BD}-\ac{RIS} and the diagonal \ac{RIS} decrease for increasing losses.

This is further investigated in Fig. \ref{fig:AnglePlotRealistic} for $\gamma = {10}^{-2}$ and a spacing of $d = 0.4$ (see Fig. \ref{fig:LossyCorner}).
In the lossless case, we can see that the increased spacing deteriorates the performance of the \ac{BD}-\ac{RIS} and increases the performance of the diagonal \ac{RIS}.
This is in accordance with Fig. \ref{fig:Corner} and the lossless curve of Fig. \ref{fig:LossyCorner}.
Hence, for this more practical spacing, the advantage of the \ac{BD}-\ac{RIS} over the diagonal \ac{RIS} is almost neglectable and they lead to roughly the same performance.
When including losses, the performance is very similar to the lossless case which is the advantage of choosing the spacing $d = 0.4 \lambda$.
In this case, the super-quadratic cannot be ensured anymore, but can still be observed around $\alpha_{\text{Rx}} = 0$ and  $\alpha_{\text{Rx}} = \pi$.

\section{Numerical Results}\label{sec:NumericalResults}
In this section, we evaluate the decoupling network together with the mutual coupling aware algorithms in the same \ac{LOS} scenario as in Section \ref{sec:ArrayPerfAnalysis} according to Fig. \ref{fig:Scenario}. 
\begin{figure}[t!]
    \includegraphics*{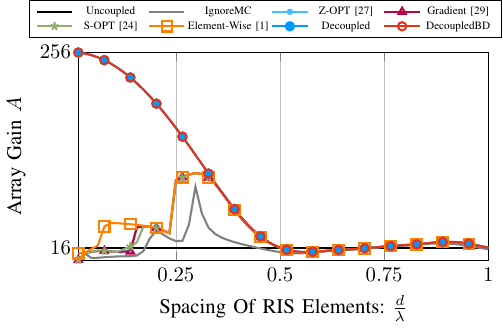}
    \caption{End-Fire direction for $N=4$.}
    \label{fig:AlgoEndFire}
\end{figure}
We focus on the diagonal \ac{RIS} where we compare the proposed decoupling network (Decoupled) with various algorithms for the conventional \ac{RIS} array under mutual coupling.
This includes the element-wise approach discussed in \cite{ConferenceDecoupling} (Element-Wise) together with the gradient approach of \cite{MutualCouplingGradient} (Gradient), the Neumann series approach of \cite{MutualCouplingZAlgo} (Z-OPT), and the two scattering parameter-based approaches in \cite{FollowUPSZJournal} (S-UNI and S-OPT).
As in \cite{FollowUPSZJournal}, S-OPT leads to better results than S-UNI in our analysis and, hence, S-UNI is omitted in the comparison.
To futher evaluate the results, we include a \ac{RIS} array without mutual coupling (Uncoupled) where we set $\bZ_\teR = \eyeM R$.
This is equal to a decoupled \ac{RIS} with a multiple $\frac{\lambda}{2}$ spacing.
Additionally, we include a naive optimization (IgnoreMC) where we solve the channel maximization by ignoring mutual coupling (assuming $\bZ_\teR = \eyeM R$) but then use the actual model with the non-diagonal coupling matrix $\bZ_\teR$ for the evaluation.
All the methods above are also initialized by this solution, i.e., by first ignoring the mutual coupling.
Please note that all baselines, i.e., Element-Wise, Gradient, Z-OPT, S-UNI and S-OPT require an iterative optimization whereas our proposed Decoupled enables closed-form solutions.
Additionally, we evaluate also the fully-connected \ac{BD}-\ac{RIS} as an upper bound.
Please note that according to Section \ref{sec:FullyConnectedRISDecoupled}, the fully-connected \ac{RIS} has the same performance with or without decoupling networks.
Hence, DecoupledBD is the performance for both versions.
In all simulations we consider the array gain as our metric.

First, we compare all methods for the end-fire direction with $N=4$ elements in Fig. \ref{fig:AlgoEndFire}.
Here, the quartic gain is achievable for $d \rightarrow 0$. 
Since, additionally, the matrix $\mathrm{Re}(\bZ_\teR)$ becomes ill-conditioned for $d \rightarrow 0$, small spacings are computationally challenging
For this scenario, the uncoupled method is significantly worse than the methods including mutual coupling and it is only well performing for $d \ge \frac{\lambda}{2}$.
The decoupling networks, on the other hand, lead to a quartic gain according to Theorem \ref{theo:MaxGain} and, additionally, Decoupled has the same performance as DecoupledBD (see Theorem \ref{theo:MaxGain}).
Moreover, the decoupling networks are clearly outperforming all mututal coupling aware algorithms for small spacings and, hence, systems without decoupling networks are not well-suited for a lossless scenario with small spacings.

\begin{figure}
    \includegraphics*{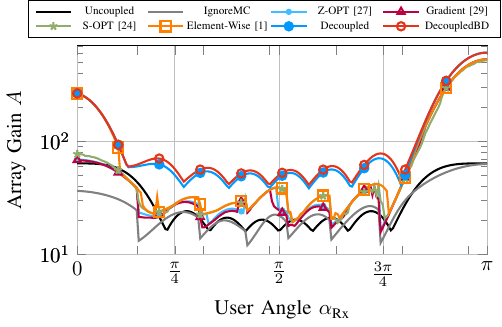}
    \caption{Element spacing of $d= 0.4 \lambda$ inlcuding losses with $\gamma={10}^{-2}$ for $N=8$.}
    \label{fig:AlgoAngleReal}
\end{figure}
Fig. \ref{fig:AlgoAngleReal} illustrates a more practical scenario with $\gamma = {10}^{-2}$ and $d=0.4$ for different user positions, i.e., angles $\alpha_{\text{Rx}}$.
We can see, that for this scenario, the DecoupledBD and the Decoupled method perform very similarly over all user angles.
Additionally, the algorithms perform better than in the theoretical case in Fig. \ref{fig:AlgoEndFire}, however, they are still worse than the Decoupled approach.
Moreover, including the effects of mutual coupling by choosing a smaller spacing than $\frac{\lambda}{2}$ increases the performance in comparison to the uncoupled method.

\section{Conclusion}
We have seen that the property of the decoupling networks to simplify the structure of the system model with mutual coupling to a structure without mutual coupling leads to new methods and new analytical results.
For lossless and very small antenna spacings, the decoupled \acp{RIS} lead to large performance gains in comparison to uncoupled \acp{RIS}.
Actually, super-quadratic gains are possible up to a quartic gain in the end-fire direction.
Also the \ac{BD}-\acp{RIS} are very promising in this ideal scenario, achieving a cubic gain over all user angles.
When, considering more realistic assumptions like lossy antenna arrays and larger element spacings (still sub-$\frac{\lambda}{2}$ spacings), the gains are lower but still  significantly better than arrays without mutual coupling and a super-quadratic gain can still be observed for certain scenarios.
For this more practical considerations, the performance of the \ac{BD}-\ac{RIS} and the diagonal \ac{RIS} with decoupling methods is very similar.
Additionally, we have seen that the decoupled \ac{RIS} outperfoms the current state-of-the-art algorithms both w.r.t. performance as well as computational complexity (closed-form solution vs. iterative algorithm).
The drawback of the decoupled networks is the increased hardware complexity which will be discussed in future work.
\vspace*{-0.5cm}
\appendix

\subsection{Global Minimum of $f(x)$}\label{app:GlobMin}
We will now derive the global minimum of the function $f_N(x)$ given by 
\begin{equation}
    f_N(x)=\summe{n=0}{N-1}(2n+1)P_n^2(x) >0
\end{equation}
for $x\in[-1,1]$ and additionally show $f_N(x) \ge \frac{N}{2} \, \forall x$.
It can be directly seen that for $x\pm 1$ the function obtains its maximum and we will exclude these points in the following.
By recognizing that $f_N(x)$ is the reciprocal of the Christoffel function for Legendre polynomials,
we can rewrite the function with the Christoffel-Darboux formula (see \cite[p.43, (3.2.4)]{PolynomialBook}) as 
\begin{equation}
    f_N(x) = N(P_N^\prime(x)P_{N-1}(x) - P_{N-1}^\prime(x)P_{N}(x)).
\end{equation}
By applying \cite[Lemma 1]{ChristoffelReformulate}, the objective is simplified to 
\begin{equation}
    f_N(x) = (1-x^2)(P_N^\prime(x))^2 + N^2 P_N^2(x)
\end{equation}
where $P_N^\prime(x)$ is the derivative of $P_N(x)$.
By recognizing similarities to the function  \cite[p.164, Theorem 7.3.1, (7.3.2)]{PolynomialBook} we define the lower bound on the objective 
\begin{equation}
    g_N(x) = \frac{N}{N+1}(1-x^2)(P_N^\prime(x))^2 + N^2 P_N^2(x) 
\end{equation}
and, hence, it holds that 
\begin{equation}
    f_N(x) \ge g_N(x).
\end{equation}
According to  \cite[p.164, Theorem 7.3.1, (7.3.2)]{PolynomialBook}, the derivative is given by 
\begin{equation}
    g_N^\prime(x) = \frac{2N}{N+1} x(P_N^\prime(x))^2.
\end{equation}
As $g_N(x)$ is monotonically increasing from $0$ to $\pm 1$, it attains its global minimum at $x=0$.
Therefore, we have 
\begin{equation}\label{eq:BoundsForChristoffel}
    f_N(x) \ge g_N(x) \ge g_N(0).
\end{equation}
\subsubsection{$N$ even}
We show now that $f_N(x) \ge \frac{N}{2}$ for even $N$.
When $N$ is even we have $P_N^\prime(0) = 0$ and, hence, 
\begin{equation}
    f_N(0)=g_N(0)=N^2 P_N^2(0).
\end{equation}
Combining this with \eqref{eq:BoundsForChristoffel}, we can already see that the global minimum of $f_N(x)$ for even $N$ is $x=0$ as $f_N(0)=g_N(0)$.
It remains to show that $P_N^2(0)\ge \frac{1}{2N}$.
From \cite[p.782, (22.7.10)]{Abramo} we have $N^2 P_N^2(0) = (N-1)^2 P_{N-2}^2(0)$ for $x=0$.
It follows that 
\begin{equation} \label{eq:SQLegendreAt0Bound}
    P_N^2(0) = \frac{(N-1)^2}{N^2}P_{N-2}^2(0) \ge \frac{(N-1)^2}{2N^2(N-2)} \ge  \frac{1}{2N}
\end{equation}
and with $P_2^2(0)=\frac{1}{4}$ the proof by induction is complete.

\subsubsection{$N$ odd}
When $N$ is odd, the function $f_N(x)$ does not attain its global minimum at $x=0$ (see Section \ref{sec:GlobMinimumDiscussion}) for a discussion).
In this case we use the lower bound, which reduces for odd $N$ to 
\begin{equation}
    f_N(x)\ge g_N(0) =\frac{N}{N+1}(P_N^\prime(0))^2
\end{equation}
as $P_N(0)=0$ for odd $N$. 
From \cite[p.783, (22.8.5)]{Abramo} we have $(P_N^\prime(0))^2 = N^2 P_{N-1}^2(0)$ for $x=0$.
Hence, it follows that 
\begin{equation}
    g_N(0) = \frac{N}{N+1} N^2 P_{N-1}^2(0) \ge \frac{N}{2}.
\end{equation}
The last step follows from the fact that $N-1$ is even and, hence, $P_{N-1}^2(0) \ge \frac{1}{2(N-1)}$ according to \eqref{eq:SQLegendreAt0Bound}.

\subsubsection{Global Minimum}\label{sec:GlobMinimumDiscussion}
We have already seen that for even $N$ the function attains its global minimum at $x=0$ with the value $f_N(0) = N^2 P_N^2(0)$.
For odd $N$ this is not the case which we will shortly discuss.
The derivative of the objective can be factorized into 
\begin{align*}
    f^\prime(x) &=2 P_N^\prime(x) b(x),\quad \text{with}
\end{align*}
\begin{align*}
    b(x) = (1-x^2)P_N^{\prime\prime}(x) -x P_N^\prime(x) + N^2  P_N(x).\\
\end{align*}
It can be seen that the extrema can be split into the two sets 
\begin{align}\label{eq:LegendreMaximumDef}
    &x_+: \, b(x_+) = 0, \, P_N^{\prime}(x_+) \neq 0 \quad \text{and}\\
    &x_-:\, P_N^{\prime}(x_-) = 0.
\end{align}
With \cite[p.781, 22.6.13; p.783, 22.8.5]{Abramo} it can be shown that 
\begin{equation}
    b(x) =  \frac{N(xP_{N-1}(x)-P_N(x))}{1-x^2}
\end{equation}
holds and, therefore, $b(x_+) = 0\hspace*{-4pt} \iff \hspace*{-4pt} P_N(x_+) = x_+P_{N-1}(x_+)$. \
Combining this with the properties of Legendre polynomials \cite[ p.783, 22.8.5; p.782, 22.7.10; p.783, 22.8.5]{Abramo}, it can be further shown that also
\begin{equation}
    f^{\prime \prime}(x_+) = -2(N-1)\frac{(P^{\prime}_{N}(x_+))^2}{1-x_+^2}<0
\end{equation}
holds and, hence, all points $x_+$ are local maxima.
As they are all maxima, the points $x_+$ don't have to be considered for the global minimum.
At the remaining stationary points $x_-$ the function can be written as 
\begin{equation}
    f_N(x_-) = N^2P_N^2(x_-)
\end{equation}
as $P_N^\prime(x_-)=0$. 
With $P_N^\prime(x_-)=0$ and $P_N(x_-) = - P_N^{\prime \prime}(x_-)\frac{1-x_-^2}{N(N+1)}$ from \cite[p.781, 22.6.13]{Abramo} it can be shown that $x_-$ are the maxima of $N^2 P_N^2(x)$ which are increasing form $0$ to $\pm 1$ (see \cite[p.164, Theorem 7.3.1, (7.3.2)]{PolynomialBook} ).
Hence, $f(x_-)$ is increasing from $0$ to $\pm 1$.
It follows that for even $N$, the global minimum is $x=0$ with the function value $f_N(0) = N^2 P_N^2(0)$.
For odd $N$, the global minimum is at $\pm x_0$ where $x_0$ is the positive zero of $P_N^\prime(x)$ nearest to $x=0$.

\vspace*{-0.018cm}
\bibliographystyle{IEEEtran}
\bibliography{refs}

\begin{thebibliography}{10}
\providecommand{\url}[1]{#1}
\csname url@samestyle\endcsname
\providecommand{\newblock}{\relax}
\providecommand{\bibinfo}[2]{#2}
\providecommand{\BIBentrySTDinterwordspacing}{\spaceskip=0pt\relax}
\providecommand{\BIBentryALTinterwordstretchfactor}{4}
\providecommand{\BIBentryALTinterwordspacing}{\spaceskip=\fontdimen2\font plus
\BIBentryALTinterwordstretchfactor\fontdimen3\font minus
  \fontdimen4\font\relax}
\providecommand{\BIBforeignlanguage}[2]{{%
\expandafter\ifx\csname l@#1\endcsname\relax
\typeout{** WARNING: IEEEtran.bst: No hyphenation pattern has been}%
\typeout{** loaded for the language `#1'. Using the pattern for}%
\typeout{** the default language instead.}%
\else
\language=\csname l@#1\endcsname
\fi
#2}}
\providecommand{\BIBdecl}{\relax}
\BIBdecl

\bibitem{ConferenceDecoupling}
D.~Semmler, J.~A. Nossek, M.~Joham, and W.~Utschick, ``{P}erformance {A}nalysis
  of {S}ystems with {C}oupled and {D}ecoupled {RIS}s,'' in \emph{2024 19th
  International Symposium on Wireless Communication Systems (ISWCS)}, 2024, pp.
  1--6.

\bibitem{Power_Min_IRS}
Q.~Wu and R.~Zhang, ``{I}ntelligent {R}eflecting {S}urface {E}nhanced
  {W}ireless {N}etwork via {J}oint {A}ctive and {P}assive {B}eamforming,''
  \emph{IEEE Transactions on Wireless Communications}, vol.~18, no.~11, pp.
  5394--5409, 2019.

\bibitem{SmartRadioEnvironment}
M.~Di~Renzo, A.~Zappone, M.~Debbah, M.-S. Alouini, C.~Yuen, J.~de~Rosny, and
  S.~Tretyakov, ``Smart radio environments empowered by reconfigurable
  intelligent surfaces: How it works, state of research, and the road ahead,''
  \emph{IEEE Journal on Selected Areas in Communications}, vol.~38, no.~11, pp.
  2450--2525, 2020.

\bibitem{EnergyEff}
C.~Huang, A.~Zappone, G.~C. Alexandropoulos, M.~Debbah, and C.~Yuen,
  ``{R}econfigurable {I}ntelligent {S}urfaces for {E}nergy {E}fficiency in
  {W}ireless {C}ommunication,'' \emph{IEEE Transactions on Wireless
  Communications}, vol.~18, no.~8, pp. 4157--4170, 2019.

\bibitem{WSR}
H.~Guo, Y.-C. Liang, J.~Chen, and E.~G. Larsson, ``{W}eighted {S}um-{R}ate
  {M}aximization for {R}econfigurable {I}ntelligent {S}urface {A}ided
  {W}ireless {N}etworks,'' \emph{IEEE Transactions on Wireless Communications},
  vol.~19, no.~5, pp. 3064--3076, 2020.

\bibitem{WMMSEMIMO}
C.~Pan, H.~Ren, K.~Wang, W.~Xu, M.~Elkashlan, A.~Nallanathan, and L.~Hanzo,
  ``{M}ulticell {MIMO} {C}ommunications {R}elying on {I}ntelligent {R}eflecting
  {S}urfaces,'' \emph{IEEE Transactions on Wireless Communications}, vol.~19,
  no.~8, pp. 5218--5233, 2020.

\bibitem{Eigenvalues}
D.~Semmler, M.~Joham, and W.~Utschick, ``{H}igh {SNR} {A}nalysis of
  {RIS}-{A}ided {MIMO} {B}roadcast {C}hannels,'' in \emph{2023 IEEE 24th
  International Workshop on Signal Processing Advances in Wireless
  Communications (SPAWC)}, 2023, pp. 221--225.

\bibitem{MIMOP2PCap}
S.~Zhang and R.~Zhang, ``{C}apacity {C}haracterization for {I}ntelligent
  {R}eflecting {S}urface {A}ided {MIMO} {C}ommunication,'' \emph{IEEE Journal
  on Selected Areas in Communications}, vol.~38, no.~8, pp. 1823--1838, 2020.

\bibitem{ZeroForcLOS}
D.~Semmler, M.~Joham, and W.~Utschick, ``{A} {Z}ero-{F}orcing {A}pproach for
  the {RIS}-{A}ided {MIMO} {B}roadcast {C}hannel,'' in \emph{ICC 2024 - IEEE
  International Conference on Communications}, 2024, pp. 4378--4383.

\bibitem{StatisticalSadaf}
S.~Syed, D.~Semmler, D.~B. Amor, M.~Joham, and W.~Utschick, ``{D}esign of a
  {S}ingle-{U}ser {RIS}-{A}ided {MISO} {S}ystem {B}ased on {S}tatistical
  {C}hannel {K}nowledge,'' in \emph{2023 57th Asilomar Conference on Signals,
  Systems, and Computers}, 2023, pp. 460--464.

\bibitem{StatisticalCSITwo}
Y.~Han, W.~Tang, S.~Jin, C.-K. Wen, and X.~Ma, ``{L}arge {I}ntelligent
  {S}urface-{A}ssisted {W}ireless {C}ommunication {E}xploiting {S}tatistical
  {CSI},'' \emph{IEEE Transactions on Vehicular Technology}, vol.~68, no.~8,
  pp. 8238--8242, 2019.

\bibitem{StatisticalCSITwoTimeScale}
M.-M. Zhao, Q.~Wu, M.-J. Zhao, and R.~Zhang, ``{I}ntelligent {R}eflecting
  {S}urface {E}nhanced {W}ireless {N}etworks: {T}wo-{T}imescale {B}eamforming
  {O}ptimization,'' \emph{IEEE Transactions on Wireless Communications},
  vol.~20, no.~1, pp. 2--17, 2021.

\bibitem{STARRIS}
J.~Xu, Y.~Liu, X.~Mu, and O.~A. Dobre, ``{STAR}-{RIS}s: {S}imultaneous
  {T}ransmitting and {R}eflecting {R}econfigurable {I}ntelligent {S}urfaces,''
  \emph{IEEE Communications Letters}, vol.~25, no.~9, pp. 3134--3138, 2021.

\bibitem{STARRISTwo}
Y.~Liu, X.~Mu, J.~Xu, R.~Schober, Y.~Hao, H.~V. Poor, and L.~Hanzo, ``{STAR}:
  {S}imultaneous {T}ransmission and {R}eflection for 360° {C}overage by
  {I}ntelligent {S}urfaces,'' \emph{IEEE Wireless Communications}, vol.~28,
  no.~6, pp. 102--109, 2021.

\bibitem{ActiveVSPassive}
Z.~Zhang, L.~Dai, X.~Chen, C.~Liu, F.~Yang, R.~Schober, and H.~V. Poor,
  ``{A}ctive {RIS} vs. {P}assive {RIS}: {W}hich {W}ill {P}revail in 6{G}?''
  \emph{IEEE Transactions on Communications}, vol.~71, no.~3, pp. 1707--1725,
  2023.

\bibitem{ActiveVSPassiveTwo}
C.~You and R.~Zhang, ``{W}ireless {C}ommunication {A}ided by {I}ntelligent
  {R}eflecting {S}urface: {A}ctive or {P}assive?'' \emph{IEEE Wireless
  Communications Letters}, vol.~10, no.~12, pp. 2659--2663, 2021.

\bibitem{ScatteringRIS}
S.~Shen, B.~Clerckx, and R.~Murch, ``Modeling and architecture design of
  reconfigurable intelligent surfaces using scattering parameter network
  analysis,'' \emph{IEEE Transactions on Wireless Communications}, vol.~21,
  no.~2, pp. 1229--1243, 2022.

\bibitem{BDRISOptimalSolution}
M.~Nerini, S.~Shen, and B.~Clerckx, ``{C}losed-{F}orm {G}lobal {O}ptimization
  of {B}eyond {D}iagonal {R}econfigurable {I}ntelligent {S}urfaces,''
  \emph{IEEE Transactions on Wireless Communications}, vol.~23, no.~2, pp.
  1037--1051, 2024.

\bibitem{BeyondDiagonal}
H.~Li, S.~Shen, and B.~Clerckx, ``{B}eyond {D}iagonal {R}econfigurable
  {I}ntelligent {S}urfaces: {F}rom {T}ransmitting and {R}eflecting {M}odes to
  {S}ingle-, {G}roup-, and {F}ully-{C}onnected {A}rchitectures,'' \emph{IEEE
  Transactions on Wireless Communications}, vol.~22, no.~4, pp. 2311--2324,
  2023.

\bibitem{NewChannelModel}
J.~A. Nossek, D.~Semmler, M.~Joham, and W.~Utschick, ``{P}hysically
  {C}onsistent {M}odeling of {W}ireless {L}inks {W}ith {R}econfigurable
  {I}ntelligent {S}urfaces {U}sing {M}ultiport {N}etwork {A}nalysis,''
  \emph{IEEE Wireless Communications Letters}, vol.~13, no.~8, pp. 2240--2244,
  2024.

\bibitem{MutualCouplingAware}
G.~Gradoni and M.~Di~Renzo, ``End-to-end mutual coupling aware communication
  model for reconfigurable intelligent surfaces: An electromagnetic-compliant
  approach based on mutual impedances,'' \emph{IEEE Wireless Communications
  Letters}, vol.~10, no.~5, pp. 938--942, 2021.

\bibitem{TowardCircuitTheory}
M.~T. Ivrlač and J.~A. Nossek, ``Toward a circuit theory of communication,''
  \emph{IEEE Transactions on Circuits and Systems I: Regular Papers}, vol.~57,
  no.~7, pp. 1663--1683, 2010.

\bibitem{FollowUpBDRIS}
H.~Li, S.~Shen, M.~Nerini, M.~Di~Renzo, and B.~Clerckx, ``{B}eyond {D}iagonal
  {R}econfigurable {I}ntelligent {S}urfaces {W}ith {M}utual {C}oupling:
  {M}odeling and {O}ptimization,'' \emph{IEEE Communications Letters}, vol.~28,
  no.~4, pp. 937--941, 2024.

\bibitem{FollowUPSZJournal}
A.~Abrardo, A.~Toccafondi, and M.~Di~Renzo, ``{D}esign of {R}econfigurable
  {I}ntelligent {S}urfaces by {U}sing {S}-{P}arameter {M}ultiport {N}etwork
  {T}heory—{O}ptimization and {F}ull-{W}ave {V}alidation,'' \emph{IEEE
  Transactions on Wireless Communications}, vol.~23, no.~11, pp.
  17\,084--17\,102, 2024.

\bibitem{FollowUpUniversal}
M.~Nerini, S.~Shen, H.~Li, M.~Di~Renzo, and B.~Clerckx, ``{A} {U}niversal
  {F}ramework for {M}ultiport {N}etwork {A}nalysis of {R}econfigurable
  {I}ntelligent {S}urfaces,'' \emph{IEEE Transactions on Wireless
  Communications}, vol.~23, no.~10, pp. 14\,575--14\,590, 2024.

\bibitem{WSAComparison}
J.~A. Nossek, D.~Semmler, M.~Joham, and W.~Utschick, ``{M}odelling of
  {W}ireless {L}inks with {R}econfigurable {I}ntelligent {S}urfaces {U}sing
  {M}ultiport {N}etwork {A}nalysis,'' in \emph{2024 27th International Workshop
  on Smart Antennas (WSA)}, 2024, pp. 1--8.

\bibitem{MutualCouplingZAlgo}
X.~Qian and M.~D. Renzo, ``{M}utual {C}oupling and {U}nit {C}ell {A}ware
  {O}ptimization for {R}econfigurable {I}ntelligent {S}urfaces,'' \emph{IEEE
  Wireless Communications Letters}, vol.~10, no.~6, pp. 1183--1187, 2021.

\bibitem{MutualCouplingSumRate}
A.~Abrardo, D.~Dardari, M.~Di~Renzo, and X.~Qian, ``{MIMO} {I}nterference
  {C}hannels {A}ssisted by {R}econfigurable {I}ntelligent {S}urfaces: {M}utual
  {C}oupling {A}ware {S}um-{R}ate {O}ptimization {B}ased on a {M}utual
  {I}mpedance {C}hannel {M}odel,'' \emph{IEEE Wireless Communications Letters},
  vol.~10, no.~12, pp. 2624--2628, 2021.

\bibitem{MutualCouplingGradient}
M.~Akrout, F.~Bellili, A.~Mezghani, and J.~A. Nossek, ``{P}hysically
  {C}onsistent {M}odels for {I}ntelligent {R}eflective {S}urface-assisted
  {C}ommunications under {M}utual {C}oupling and {E}lement {S}ize
  {C}onstraint,'' \emph{arXiv 2302.11130}, 2023.

\bibitem{ElementWise}
H.~E. Hassani, X.~Qian, S.~Jeong, N.~S. Perović, M.~Di~Renzo, P.~Mursia,
  V.~Sciancalepore, and X.~Costa-Pérez, ``{O}ptimization of {RIS}-{A}ided
  {MIMO}—{A} {M}utually {C}oupled {L}oaded {W}ire {D}ipole {M}odel,''
  \emph{IEEE Wireless Communications Letters}, vol.~13, no.~3, pp. 726--730,
  2024.

\bibitem{BDRISCoupling}
\BIBentryALTinterwordspacing
M.~Nerini, H.~Li, and B.~Clerckx, ``{G}lobal {O}ptimal {C}losed-{F}orm
  {S}olutions for {I}ntelligent {S}urfaces {W}ith {M}utual {C}oupling: {I}s
  {M}utual {C}oupling {D}etrimental or {B}eneficial?'' 2024. [Online].
  Available: \url{https://arxiv.org/abs/2411.04949}
\BIBentrySTDinterwordspacing

\bibitem{SquaredArrayGain}
E.~Altshuler, T.~O'Donnell, A.~Yaghjian, and S.~Best, ``{A} monopole
  superdirective array,'' \emph{IEEE Transactions on Antennas and Propagation},
  vol.~53, no.~8, pp. 2653--2661, 2005.

\bibitem{PolynomialBook}
G.~Szeg\"{o}, \emph{{O}rthogonal {P}olynomials}.\hskip 1em plus 0.5em minus
  0.4em\relax Providence, RI: American Mathematical Society, 1975.

\bibitem{ChristoffelReformulate}
L.~Bos, , A.~Narayan, N.~Levenberg, and F.~Piazzon, ``{A}n {O}rthogonality
  {P}roperty of the {L}egendre {P}olynomials,'' \emph{Constructive
  Approximation}, 2017.

\bibitem{Abramo}
M.~Abramowitz and I.~A. Stegun, Eds., \emph{Handbook of Mathematical Functions
  with Formulas, Graphs and Mathematical Tables}.\hskip 1em plus 0.5em minus
  0.4em\relax New York: Dover Publications, Inc., 1965.

\end{thebibliography}

\end{document}